\documentclass[a4paper,11pt]{article}
\usepackage{amsmath,a4wide}
\usepackage{amssymb}
\usepackage{amscd}
\usepackage{bm}
\usepackage{epsfig}
\usepackage{graphics}
\usepackage{graphicx}
\usepackage{float}
\usepackage{epsfig}
\usepackage{amsmath}
\usepackage{graphics}
\usepackage{graphicx}
\usepackage{color}
\usepackage{epstopdf}
\usepackage{color}
\usepackage{booktabs}
\usepackage{subfigure}
\usepackage{ifpdf}
\usepackage{amssymb}
\usepackage{cite}
\usepackage{amssymb}
\usepackage{amsthm}
\usepackage{amsmath}
\usepackage[justification=centering]{caption}
\usepackage[]{natbib}
\usepackage[titletoc]{appendix}
\newtheorem{definition}{Definition}
\newtheorem{theorem}{Theorem}

\newtheorem{pro}{Proposition}
\newtheorem{cor}{Corollary}
\newtheorem{remark}{Remark}
\begin{document}
	
   \title{Volterra mortality model: Actuarial valuation and risk management with long-range dependence}

\author{Ling Wang\thanks{Department of Statistics, The Chinese University of Hong Kong, Shatin, N.T., Hong Kong. \newline({\tt lingwang@link.cuhk.edu.hk})}
	 \and Mei Choi Chiu\thanks{Department of Mathematics \& Information Technology, The Education University of Hong Kong, Tai Po, N.T., Hong Kong. \newline({\tt mcchiu@eduhk.hk})}	 \and Hoi Ying Wong\thanks{Corresponding author. Department of Statistics, The Chinese University of Hong Kong, Shatin, N.T., Hong Kong. \newline({\tt hywong@cuhk.edu.hk})}
 }
\date{\today}

    \maketitle \pagestyle{plain} \pagenumbering{arabic}

\abstract{While abundant empirical studies support the long-range dependence (LRD) of mortality rates, the corresponding impact on mortality securities are largely unknown due to the lack of appropriate tractable models for valuation and risk management \textcolor{black}{purposes}. We propose a novel class of Volterra mortality models that incorporate LRD into the actuarial valuation, retain tractability, and are consistent with the existing continuous-time affine mortality models. We derive the survival probability in closed-form solution by taking into account of the historical health records. The flexibility and tractability of the models make them useful in valuing mortality-related products such as death benefits, annuities, longevity bonds, and many others, as well as offering optimal mean-variance mortality hedging rules.  Numerical studies are conducted to examine the effect of incorporating LRD into mortality rates on various insurance products and hedging efficiency.}

	\vspace{2mm} {\em Keywords} : Stochastic mortality; Long-range dependence; Affine Volterra processes; Valuation; Mean-variance hedging. \\
	
\clearpage
\section{Introduction}
Actuaries heavily rely on mortality modeling for mortality prediction, actuarial valuation, and risk management. Accurate estimations and predictions of human mortality are the essential building blocks of both insurance contract pricing and pension policy. The first study of this can be dated back to
\cite{GB}. 

The arguably most well-received modern mortality model is the \cite{LC} model and its extensions using  time series analysis. For instance, it has been generalized to multivariate populations with a common trend
\citep{LL}, mortality forecasts using single value decomposition \citep{RH}, joint modeling of different national populations \citep{ABO} and sub-populations \citep{VH}, a multi-population stochastic mortality model \citep{DHM}, a Poisson regression model \citep{BDV}, and stochastic period and cohort effect \citep{TPF}, among others.  A key advantage of the Lee-Carter model and its \textcolor{black}{invariant} is that statistical inferences from time series analysis can be applied or generalized to estimate and test with a real mortality data set. 

By incorporating fractionally integrated time series analysis into the Lee-Carter model, \cite{YPC} empirically show the existence of long-range dependence (LRD) (also known as long-memory pattern or fractional persistence) across age groups, gender, and countries by using the dataset of 16 countries. When they apply their long-memory mortality model to forecast life expectancies, the mortality model ignoring LRD tends to underestimate life expectancy, which leads to important implications for pension schemes and funding issues.  \cite{YPCM} further extend the model to incorporate multivariate cohorts and document the existence of LRD. \cite{YGA} show a long-memory pattern in the infant mortality rates of G7 countries. \cite{DO} offer further empirical evidence that mortality rates exhibit LRD using a fractional Ornstein-Uhlenbeck (fOU) process with Italian population data from the 1950 to 2004 period.

Most stochastic mortality models focus on the mortality rate, or equivalently the Poisson intensity rate. We refer to the pioneering work of \cite{MP} who introduced the Cox model to insurance applications.  \cite{Biffis} and \cite{BM} further develop this idea of doubly stochastic mortality models with an affine feature for exploiting analytical tractability in actuarial valuation with both financial and mortality risks. \cite{JPE} extend it to cohort models, and \cite{WCW} introduce continuous-time cointegration into the multivariate mortality rates. 

\cite{BS} advocate the use of continuous-time affine mortality  models for longevity pricing and hedging because of its tractability and consistency with the market data. \cite{JP} propose a calibration to the multiple populations affine mortality models and demonstrate its empirical use with product price data. However, none of the aforementioned studies provide an analytically tractable dynamic mortality model with the LRD feature.

The primary contribution of this paper is the proposal of a novel class of dynamic stochastic mortality models that simultaneously render actuarial valuation tractability and the LRD property. As the proposed model is based on Volterra processes, we call them Volterra mortality models. Inspired by the affine Volterra process \citep{Jaber}, our model preserves the affine structure for general actuarial valuation but still captures LRD. In terms of practical contributions, we use the model to derive closed-form solutions for the survival probability, death and survival benefits of insurance contracts, and longevity bonds, and then address the impact of LRD on these insurance products. To the best of our knowledge, the derived formulas constitute the first set of formulas for insurance products that are subject to the LRD feature of mortality rates. 

This study also contributes to risk management with LRD mortality rates. We rigorously develop the mean-variance (MV) strategy for hedging longevity risk with a longevity security that is subject to LRD. This later hedging strategy is highly non-trivial because the Volterra mortality rate is a non-Markovian and non-semimartingale process. Inspired by \cite{HW}, we derive the MV optimal hedging with the Volterra mortality models by means of linear-quadratic control with the backward stochastic differential equation (BSDE) framework similar to \cite{WCW}. In contrast, \cite{HW} solve the MV portfolio problem with rough volatility by constructing an auxiliary process. Our optimal hedging rule shows how to adjust the hedge for LRD of mortality rates.

The rest of this paper is organized as follows. Section \ref{Model} introduces the Volterra mortality model based on the doubly stochastic mortality models and explains how the model captures LRD.  Section \ref{Valuation} offers some formulas for actuarial valuation.   In Section \ref{Hedge}, we formulate an optimal hedging problem under the Volterra mortality model and give an explicit solution. To compare the Volterra mortality model with LRD with the Markovian mortality model, numerical studies are conducted for both actuarial valuation and the hedging problem in Section \ref{numerical}. Section \ref{Conclusion} gives our concluding remarks. Some details and additional proofs are given in the Appendix. 

\section{The model} \label{Model}
Consider a filtered probability space $(\Omega, \mathcal{F},\mathbb{F},\mathbb{P})$ where the filtration $\mathbb{F}= \{\mathcal{F}_t:0\leq t \leq T\}$ satisfies the usual properties. We write $\mathcal{F}_t=\mathcal{G}_t\vee\mathcal{H}_t$, where $\mathcal{H}_t$ represents the flow of information available as time goes by including the historical processes and the current states, and $\mathcal{G}_t$ contains the information whether an individual has died. We interpret $\mathbb{P}$ as the physical probability measure. Alternatively, our model can be developed under a pricing measure so that the model parameters are calibrated to the insurance product prices available in the market. This enables actuarial valuation consistent with market prices. However, risk management strategies should be conducted under the physical probability measure. To avoid confusion, we denote the pricing measure by $\mathbb{Q}$ and discuss the relationship between $\mathbb{P}$ and $\mathbb{Q}$ in the next section. For the time being, we focus on the model development under $\mathbb{P}$.

We begin with the classic doubly stochastic mortality models. For simplicity, we consider a group of people with homogeneous feature while individual differences certainly exist in this group at the same time. A counting process $N$ is a doubly stochastic process driven by the subfiltration $\mathbb{G} = \{{\mathcal G}_t\}_{t\ge 0}$ of $\mathbb{F}$ and with $\mathbb{G}$-intensity $\mu_t$. Let $\tau$ be the first jump-time of the process $N$ with intensity $\mu_t$. In actuarial applications, the process $\{N_t\}_{t\ge 0}$ records the number of deaths at each time $t\geq 0$. For any time $t\geq 0$ and state $\omega \in \Omega$ such that $\tau(\omega)>t$, we have
\begin{eqnarray}
\mathbb{P}(\tau\leq t+\Delta|\mathcal{F}_t)\cong \mu_t(\omega)\Delta,
\end{eqnarray}
for a trajectory of $\mu_t(\omega)$ and a fixed $\omega\in\Omega$. Thus, the counting process $N$ associated with $\tau$ becomes an inhomogeneous Poisson with parameter $\int_{0}^{\cdot}\mu_s(\omega)ds$. In other words, for all $T\geq t \geq 0$ and integer $k$ ($k\geq 0$), we have
\[\mathbb{P}(N_T-N_t=k|\mathcal{F}_t\vee \mathcal{G}_T) = \frac{(\int_{t}^{T}\mu_s(\omega)ds)^k}{k!}e^{-\int_{t}^{T}\mu_s(\omega)ds}.\]
By the law of iterated expectations, the time-$t$ survival probabilities over the time interval $(t,T]$ (for fixed $T\geq t \geq 0$) can be expressed as follows:
\begin{equation}\label{Conditional S}
	\mathbb{P}(\tau>T|\mathcal{F}_t)=\mathbb{E}\left[\left.e^{-\int_{t}^{T}\mu_s(\omega)ds}\right|\mathcal{F}_t\right].
\end{equation}
If the intensity $\mu_t$ is a constant, then the doubly stochastic process reduces to the homogeneous Poisson process. However, the literature of mortality modeling is in favour of a stochastic intensity. Typically, the intensity is modeled through a stochastic differential equation (SDE). For instance, \cite{Biffis} and \cite{BM} postulate a Markovian process such that $\mu_t = f(X_t)$, where $f$ is a continuous function on $\mathbb{R}$,
\begin{eqnarray}
dX_t =  b(X_t)dt + \sigma(X_t)dW_t, \label{SDE}
\end{eqnarray}
and $\{W_t\}_{t\ge 0}$ is the standard Brownian motion. 

To incorporate LRD into the mortality rate, one simply replaces the Brownian motion in \eqref{SDE} with the fractional Brownian motion. In other words,
\begin{eqnarray}
dX_t =  b(X_t)dt + \sigma(X_t)dW_t^H, \label{fSDE}
\end{eqnarray}
where $W_t^H$ is a fractional Brownian motion (fBM) with the Hurst parameter $H \in [0.5, 1)$. For instance, the empirical study of \cite{DO} uses $f(X_t) =h_0{\rm exp}(h_1t+h_2X_t)$, for the constants $h_0,h_1, h_2>0$, and a fOU process in the form of \eqref{fSDE} such that the drift term $b(X_t)$ is a linear function of $X_t$ and the $\sigma(X_t) \equiv \sigma$ is a constant. However, the fractional Brownian motion is analytically intractable for actuarial valuation.

\subsection{Volterra mortality}
We propose a stochastic mortality model incorporating LRD that retains the key advantages of the works of \cite{Biffis}, \cite{DO}, and \cite{LST}. More specifically, we maintain the affine nature of \cite{Biffis}, reflect LRD with fBM as in \cite{DO}, and offers explicit expressions for some important Fourier-Laplace functional generalizing \cite{LST} for actuarial valuation. Our model is highly inspired by the affine Volterra processes \citep{Jaber} and hence called the Volterra mortality model.

In the one dimensional case, \cite{BN} show the equivalence between fBM and the Volterra process:
$$W_t^H = c_H\int_0^t (t-s)^{H-\frac{1}{2}}dW_1(s),$$
where $c_H$ is a constant related to the Hurst parameter $H$, $W_1$ is the Wiener process, and the integral process on the right-handed side is a standard Volterra process. For simplicity and to be consistent with the literature, we postulate the mortality rate $\mu_t$ of a group: 
\begin{equation}
\mu_t  = m(t) + \eta X_t, \label{mu0}
\end{equation}
where $m(t)$ is a bounded continuous deterministic function and $\eta$ is a constant. In other words, we require that $f(X_t)$ is a linear function of $X_t$. In addition, $X_t$ follows a stochastic Volterra integral equation (SVIE): 
\begin{equation}\label{volterra process}
   X_t = X_0 +\int_0^tK(t-s)b(X_s)ds +\int_0^t K(t-s)\sigma(X_s)dW_s,
\end{equation}
where $W = [W_1, \cdots, W_d]^\top$ is the standard $d$-dimensional Brownian motion under $\mathbb{P}$, and the coefficients $b$ and $\sigma$ are assumed to be continuous.  \textcolor{black}{The convolution kernel $K$ satisfies the following condition:
\begin{itemize}
\item[]$K \in L^2_{loc}(\mathbb{R}_+, \mathbb{R})$, $\int_{0}^{h}K(t)^2dt = O(h^\gamma)$ and $\int_{0}^{T}(K(t+h)- K(t))^2dt = O(h^\gamma)$ for some $\gamma \in (0, 2]$ and  every $T < \infty$. 
\end{itemize}
}
Although the process $X_t$ in \eqref{volterra process} is generally high-dimensional, we would like to illustrate it in a one-dimensional case. Table \ref{kernel} exhibits some useful kernels $K$ in the one-dimensional case. We obtain the fBM by choosing $K$ as the fractional kernel in Table \ref{kernel} with a constant $\sigma(X_s)$ and $b=0$ in \eqref{volterra process}. Therefore, the Volterra processes can be applied to a wider class of LRD noise terms. Note that the \textit{resolvent} or \textit{resolvent of the second kind} corresponding to the $K$ shown in Table \ref{kernel} is defined as the kernel $R$  such that $K*R=R*K=K-R$. The convolutions $K*R$ and $R*K$ with $K$ a measurable function on $\mathbb{R}_+$ and $R$ a measure on $\mathbb{R}_+$ of locally bounded variation are defined by 
 \[(K*R)(t) = \int_{[0,t]}K(t-s)R(ds),~~ (R*K)(t) = \int_{[0,t]}R(ds)K(t-s)\]
 for $t>0$.
\textcolor{black}{\begin{remark}
	  According to \cite{Biffis}, the deterministic function $m(t)$ in \eqref{mu0} may represent (i) a best-estimated assumption on $\mu$ enforcing unbiased expectations about the future based on the available information, (ii) pricing demographics basis, or (iii) an available mortality table for a population of insureds. In Section \ref{numerical}, we calibrate $m(t)$ to the table SIM92, a period table usually employed to price assurances.
\end{remark}}
 
\begin{table}[H]
	\centering
	\begin{tabular}{ccccc}  
		\toprule
		\toprule
		&Constant&Fractional&Exponential&Gamma \\
		\midrule
		\midrule
		$K(t)$& $c$ & $c\frac{t^{\alpha-1}}{\Gamma(\alpha)}$ & $ce^{-\lambda t}$  & $ce^{-\lambda t}\frac{t^{\alpha-1}}{\Gamma(\alpha)}$\\
		\midrule
		$R(t)$& $ce^{-ct}$ & $ct^{\alpha-1}E_{\alpha,\alpha}(-ct^{\alpha})$&$ce^{-\lambda t}e^{-ct}$ & $ce^{-\lambda t}t^{\alpha-1}E_{\alpha,\alpha}(-ct^{\alpha})$\\
		\bottomrule
	\end{tabular}
	\caption{Examples of kernel function $K$ and the corresponding resolvent $R$. Here $E_{\alpha, \beta}(z)=\sum_{n=0}^{\infty}\frac{z^{n}}{\Gamma(\alpha n+\beta)}$ denotes the Mittag-Leffler function.}
	\label{kernel}
\end{table}

 In addition, when the convolution kernel $K$ is set to a constant $c$ in \eqref{volterra process}, the $X_t$ reduces to the solution of a SDE. Furthermore, once $b(X_t)$ is linear in $X_t$ and $\sigma(X_t)$ satisfies a certain affine property, then our model in \eqref{volterra process} becomes the affine stochastic mortality model of \cite{Biffis}. The possibly high-dimensional $X_t$ enables us to also incorporate multi-factor mortality modeling. However, we would like to highlight that the Volterra process in \eqref{volterra process} is generally a non-Markovian and non-semimartingale process. The non-Markovian nature is obvious because the integrals in the SIVE take the whole realized sample path into account. The non-semimartingale feature is reflected by the fact that the time variable $t$ appears in both the integral limit and the kernel function, making it fail to define the It\^o integral.

Fortunately, \cite{Jaber} show that it is still possible to maintain the affine nature within \eqref{volterra process}. Let $a(x)=\sigma(x)\sigma(x)^\top$ be the covariance matrix. 
\begin{definition}\label{def1}
The SVIE \eqref{volterra process} is called an affine process \citep{Jaber} if
\[a(x)=A^0+x_1A^1+\cdot+x_dA^d,\]
\[b(x)=b^0+x_1b^1+\cdots+x_db^d,\]
for some $d$-dimensional symmetric matrices $A^i$ and vectors $b^i$. For simplicity, we set $B=(b^1,\cdots, b^d)$ and $A(u)=(uA^1u^\top,\cdots, uA^du^\top)$ for any row vector $u \in \mathbb{C}^d$.
\end{definition}
To draw insights from Definition \ref{def1}, consider the one dimensional case. When $b(x) = b^0 - b^1x$, a linear function of $x$, and $a(x)$ is a constant, \eqref{volterra process} is known as the Volterra type of the {\bf Vasicek} (VV) model which reduces to the classic Vasicek model by taking a constant kernel or, equivalently, $H = 1/2$ in the fractional kernel. 
When $b(x)$ is linear in $x$ and $a(x)$ is directly proportional to $x$, our model in \eqref{volterra process} reduces to the Volterra version of the {\bf CIR} (VCIR) model. 

\subsection{Interest rate model}
Although we focus on mortality modeling, actuarial valuation needs to specify the dynamic of the risk-free interest rate. We simply adopt a Markov affine model for the interest rate.  
Specifically, we adopt the short rate process $r$ that satisfies $\int_{0}^{t}|r_s|ds<\infty$ for $t\geq 0$, and we define the return of a risk-less asset as ${\rm exp}(\int_{0}^{t}r_sds)$ for a unit dollar investment at time 0. In addition, the interest rate process is driven by the Markov affine process $Z$ in $\mathbb{R}^k$:
\begin{eqnarray}
dZ_t=\widetilde{b}(Z_t)dt+\widetilde{\sigma}(Z_t)dW'_t, \label{Z}
\end{eqnarray}
where $W'$ is a $k$-dimensional standard Brownian motion.  The coefficients $\widetilde{b}(Z_t)$ and $\widetilde{a}(Z_t) = \widetilde{\sigma}(Z_t)\widetilde{\sigma}^\top(Z_t)$ have affine dependence on $Z_t$ once they satisify Definition \ref{def1} with the dimension $d$ replaced by $k$. Hence, the Markov affine feature coincides with the definition of Markov affine process in \cite{DFS}.  Furthermore, the short rate $r_t\doteq r(t,Z_t)=\lambda_0(t)+\lambda_1(t)\cdot Z_t$  which is an affine function on $Z_t$ with coefficients $\lambda_0(t)$ and $\lambda_1(t)$ being bounded continuous functions on $[0,\infty)$. By the affine processes in \cite{DFS} and  \cite{F}, at time $t$, we have
\begin{equation}\label{affineMarkov}
\mathcal{B}(t,T)=\mathbb{E}\left[\left.e^{-\int_{t}^{T}r(s,Z_s)ds}\right|\mathcal{F}_t\right]=e^{\widetilde{\alpha}(t, T)+\widetilde{\beta}(t, T)\cdot Z_t},
\end{equation} 
where the functions $\tilde{\alpha}(\cdot, T)$ and $\tilde{\beta}(\cdot,  T)$ are uniquely solved from the ordinary differential equations (ODEs) in Appendix \ref{appendix:affine} with boundary conditions $\widetilde{\alpha}(T, T)=0$ and $\widetilde{\beta}(T, T)=0$. If the interest rate model in \eqref{Z} is defined under the pricing measure, i.e., $\mathbb{P} = \mathbb{Q}$, then the quantity $\mathcal{B}(t,T)$ represents the price of a unit zero coupon bond.

\section{Actuarial Valuation}\label{Valuation}
We demonstrate the tractability of the proposed Volterra mortality model in actuarial valuation. Specifically, we derive closed-form solutions to the survival probability and prices of some standard life insurance products.  The following theorem is the building block of the actuarial valuation.
\begin{theorem}\label{theorem1}
	If the mortality rate $\mu_t$ follows \eqref{mu0} and \eqref{volterra process} and has the affine structure specified in Definition \ref{def1}, then, for any constant $c_0$ and $c_1$ and $T>t$, we have
	\begin{align}\label{eq}
		\mathbb{E}\left[\left.e^{-\int_{t}^{T}\mu_sds}(c_0+c_1\mu_T)\right|\mathcal{F}_t^X\right]=	c_0g(t,T)- c_1\frac{\partial g(t, T)}{\partial T},
	\end{align}
where
\begin{eqnarray}
g(t, T)&=&e^{-\int_{0}^{T}m(s)ds}e^{\int_{0}^{t}\mu_sds}{\rm exp}(Y_t(T)),\nonumber\\ 
	Y_t(T)&=&Y_0+\int_{0}^{t}\psi(T-s)\sigma(X_s)dW_s-\frac{1}{2}\int_{0}^{t}\psi(T-s)a(X_s)\psi(T-s)^\top ds, \label{Y} \\
	Y_0(T) &=& \int_{0}^{T}(-\eta X_0+\psi(s)b(X_0)+\frac{1}{2}\psi(s)a(X_0)\psi(s)^\top)ds, \nonumber
	\end{eqnarray}
and $\psi\in {\mathcal L}^2([0,T], \mathbb{C}^d)$ solves the Riccati-Volterra equation:
	\begin{align}\label{riccati}
		\psi=(-\eta+\psi B+\frac{1}{2}A(\psi))*K,
	\end{align}
with $A(\cdot)$ appearing in Definition \ref{def1}.
In addition, the $Y$ has an alternative expression:
\begin{equation}
Y_t(T) = -\eta\int_{0}^{T}\mathbb{E}[X_s|\mathcal{F}_t]ds+ \frac{1}{2}\int_{t}^{T}\psi(T-s)a(\mathbb{E}[X_s|\mathcal{F}_t])\psi(T-s)^\top ds, \label{Y2}
\end{equation}
where	
\begin{small}
\begin{equation}\label{expectation}
	\mathbb{E}[X_T|\mathcal{F}_t]=\left({\rm id} - \int_{0}^{T}R_B(s)ds\right)X_0 +\int_{0}^{T}E_B(T-s)b^0(s)ds+\int_{0}^{t}E_B(T-s)\sigma(X_s)dW_s
\end{equation}
\end{small}
with ${\rm id}$ being the identity matrix, $R_B$ the resolvent of $-KB$, and $E_B = K-R_B*K$. 
\end{theorem}
\begin{proof}
	See Appendix \ref{appendix:affine}. 
\end{proof}

\begin{remark}\label{remark1}
The partial derivative $\frac{\partial g(t, T)}{\partial T}$ does not admit a closed-form solution in general because the function $g(t,T)$ depends on $Y_t(T)$ which depends on $T$ through the $\psi$ solved from the Riccati-Volterra Equation \eqref{riccati}. Fortunately, the partial derivative appears in insurance products related to the death benefit through an integration. We can then avoid computing it by means of integration by parts.
\end{remark}

We highlight that the expression in \eqref{Y} implies that $Y_t(T)$ is a semimartingale, because all of the integrants in \eqref{Y} are independent of $t$. This is important and interesting because it implies that insurance product prices can be expressed into SDE even though the mortality rate with LRD can not. This enables us to construct a hedging strategy for longevity risk using longevity securities in a LRD mortality environment, indicating the importance of the longevity securatization. For the time being, we apply Theorem \ref{theorem1} to obtain the survival probability of the Volterra mortality model in a closed-form solution.

\begin{cor}\label{cor1}
 (Survival Probability) Under the Volterra mortality model in \eqref{mu0}, \eqref{volterra process}, and Definition \ref{def1}, for any $t<T$, the survival probability reads
\begin{align}\label{survival}
\mathbb{P}(\tau>T|\mathcal{F}_t)=\mathbb{E}\left[\left.e^{-\int_{t}^{T} \mu_sds}\right|\mathcal{F}_t\right]=g(t,T) = e^{-\int_{0}^{T}m(s)ds+\int_0^t\mu_sds}\exp(Y_t(T)),
\end{align}
where  $Y_t(T)$ is defined in \eqref{Y} or, equivalently, \eqref{Y2}.
\end{cor}
\begin{proof}
The result follows by taking $c_0 = 1$ and $c_1 = 0$ in Theorem \ref{theorem1}.
\end{proof}

The survival probability in Corollary \ref{cor1} captures LRD because it depends on the whole historical path of the mortality rate. This is reflected in the terms $e^{-\int_{0}^{T}m(s)ds+\int_0^t\mu_sds}$
and $Y_0(T)$. However, when comparing our survival probability with LRD with that of the corresponding Markovian mortality model, we find them consistent. Consider the case of fractional kernel $K(t)=\frac{t^{\alpha-1}}{\Gamma(\alpha)}{\rm id}$, where $\alpha = H + 1/2$ and $H$ is the Hurst parameter $H$. The process $X_t$ becomes
\begin{equation}\label{fractional X}
X_t = X_0 +\lambda\int_0^t\frac{(t-s)^{\alpha-1}}{\Gamma(\alpha)}(\theta-X_s)ds +\int_0^t \frac{(t-s)^{\alpha-1}}{\Gamma(\alpha)}\sigma(X_s) dW_s.
\end{equation}
When $\alpha=1$, the $K(t)\equiv {\rm id}$ and
\[dX_t=\lambda(\theta-X_t)dt + \sigma(X_t) dW_t,\]
which is the Vasicek mortality rate model for a constant $\sigma(X_t)$ and the CIR model for $\sigma(X_t) = \sigma\sqrt{X_t}$. Both are investigated by \cite{Biffis}. In such a situation, a part of the $Y_0(T)$ in \eqref{Y} cancels with $\int_0^t\mu_sds$, and the Volterra-Riccati Equation \eqref{riccati} reduces to the ordinary Riccati equation. This makes our solution the same as these in \cite{Biffis} for $\alpha = 1$ or $H = 1/2$. However, once $\alpha>1$, the process $X_t$ has the LRD feature.  The empirical study in \cite{YPC} shows that the survival probability is underestimated when LRD is not taken into account. 

\subsection{Standard Insurance contracts}
To streamline the presentation, we assume that mortality rates are independent of the interest rate. Although this assumption could be considered as mathematically restrictive, it is a common assumption in the actuarial and insurance literature. Two basic payoffs in insurance contracts are the survival benefit and the death benefit.

Let $C_T$ be a bounded random payoff for a survivor at time $T$ independent of the mortality. The time-$t$ fair value of the survival benefit ${\rm SB}_t(C_T; T)$ of the terminal amount $C_T$, with $0\leq t \leq T$  under the pricing measure $\mathbb{Q}$ is given by
\begin{equation}
{\rm SB}_t(C_T; T)= 1_{\{\tau>t\}}\mathbb{E}^{\mathbb{Q}}\left[\left.e^{-\int_{t}^{T}r_sds}C_T\right|\mathcal{G}^Z_t\right]\mathbb{E}^{\mathbb{Q}}\left[\left.e^{-\int_{t}^{T}\mu_sds}\right|\mathcal{G}_t^X\right]. \label{SBpayoff}
\end{equation}
To draw some insights from \eqref{SBpayoff}, let us consider the situation in which the mortality model of \eqref{mu0} and \eqref{volterra process} and interest rate process of \eqref{Z} are constructed under the pricing measure $\mathbb{Q}$ or, equivalently, that $\mathbb{P} = \mathbb{Q}$ in Section \ref{Model}. We refer to the results obtained under such an assumption as the baseline case in this paper and the corresponding valuation becomes simple.

\begin{pro}\label{pro1}
 (Survival Benefit: The Baseline Valuation.) If $\mathbb{P} = \mathbb{Q}$ and the mortality and interest rate are independent, then the Volterra mortality model of \eqref{mu0}, \eqref{volterra process}, and Definition \ref{def1} and the affine interest rate model imply that
\[{\rm SB}_t(C_T; T)= 1_{\{\tau>t\}}\mathcal{B}(t,T)\mathbb{E}^{\mathbb{Q}^T}\left[\left.C_T\right|\mathcal{G}^Z_t\right]g(t,T),\]
where  $g(t,T)$ is presented in Theorem \ref{theorem1}, $\mathcal{B}(t,T)$ is the zero coupon bond price in \eqref{affineMarkov}, and $\mathbb{Q}^T$ is the forward pricing measure:
$$\left.\frac{d\mathbb{Q}^T}{d\mathbb{Q}}\right|_{\mathcal{F}_t} = \exp\left(-\frac{1}{2}\int_0^t \widetilde{\beta}^2(u, T)\widetilde{\sigma}^2(Z_u)\,du - \int_0^t \widetilde{\beta}(u, T)\widetilde{\sigma}(Z_u)dW'_u \right).$$
\end{pro}
\begin{proof}
By Corollary \ref{cor1},
$$\mathbb{E}^{\mathbb{Q}}\left[\left.e^{-\int_{t}^{T}\mu_sds}\right|\mathcal{G}_t^X\right] = g(t,T).$$
By the affine short-rate Model \eqref{Z} and Equation \eqref{affineMarkov}, we have
\[d\mathcal{B}(t,T) = \mathcal{B}(t,T)r_tdt - \mathcal{B}(t,T)\widetilde{\beta}(t,T)\widetilde{\sigma}(Z_t)dW'_t,\]
which implies that $1= \mathcal{B}(T,T) =  \mathcal{B}(t,T)e^{\int_t^T r_u-\frac{1}{2} \widetilde{\beta}^2(u, T)\widetilde{\sigma}^2(Z_u)\,du - \int_t^T \widetilde{\beta}(u, T)\widetilde{\sigma}(Z_u)dW'_u }$. Hence,
$$e^{-\int_{t}^{T}r_sds} = \mathcal{B}(t,T)\exp\left(-\frac{1}{2}\int_t^T \widetilde{\beta}^2(u, T)\widetilde{\sigma}^2(Z_u)\,du - \int_t^T \widetilde{\beta}(u, T)\widetilde{\sigma}(Z_u)dW'_u \right).$$
An application of the Girsanov theorem shows that
$$\mathbb{E}^{\mathbb{Q}}\left[\left.e^{-\int_{t}^{T}r_sds}C_T\right|\mathcal{G}^Z_t\right] = \mathcal{B}(t,T)\mathbb{E}^{\mathbb{Q}^T}\left[\left.C_T\right|\mathcal{G}^Z_t\right],$$
where the forward measure $\mathbb{Q}^T$ is presented in the Proposition.
\end{proof}

Another important basic payoff is the death benefit. Let $C_t$ be a bounded $\mathbb{G}^Z$-predictable process, representing a cash flow stream independent of the mortality rate. Then, the time-$t$ fair value of the death benefit with a cash flow stream $C_{t}$, payable in case the insured dies before time $T$ and $0\leq t \leq T$, is given by
\[{\rm DB}_t(C_{\tau}; T)=1_{\{\tau >t\}}\int_{t}^{T}\mathbb{E}^{\mathbb{Q}}\left[\left.e^{-\int_{t}^{u}r_sds}C_u\right|\mathcal{G}_t^Z\right]\mathbb{E}^{\mathbb{Q}}\left[\left.e^{-\int_{t}^{u}\mu_sds}\mu_u\right|\mathcal{G}_t^X\right]du.\]
Then, we also have an explicit baseline valuation formula for the death benefit.
\begin{pro}\label{pro2}
 (Death Benefit: The Baseline Valuation.) If $\mathbb{P} = \mathbb{Q}$ and the mortality and interest rate are independent, then the Volterra mortality model of \eqref{mu0}, \eqref{volterra process}, and Definition \ref{def1} and the affine interest rate model imply that
\begin{align}
{\rm DB}_t(C_T; T) = -1_{\{\tau>t\}}\int_{t}^{T}\mathcal{B}(t,u)\mathbb{E}^{\mathbb{Q}^u}\left[\left.C_u\right|\mathcal{G}^Z_t\right]\frac{\partial g(t, u)}{\partial u}du, \notag
\end{align}
where  $Y_t(u)$ is defined in \eqref{Y}, $\mathcal{B}(t,T)$ in \eqref{affineMarkov}, $g(t,T)$ in Theorem \ref{theorem1}, and the forward pricing measure $\mathbb{Q}^u$ in Proposition \ref{pro1}.
\end{pro}
\begin{proof}
The proof is similar to that of Proposition \ref{pro1} except for the second expectation appearing in the representation of ${\rm DB}_t(C_{\tau}; T)$. By Theorem \ref{theorem1}, it is clear that
\[\mathbb{E}^{\mathbb{Q}}\left[\left.e^{-\int_{t}^{u}\mu_sds}\mu_u\right|\mathcal{G}_t^X\right] = -\frac{\partial g(t, u)}{\partial u}.\]
\end{proof}
Applying integration by parts to DB in Proposition \ref{pro2} yields an alternative expression:
\begin{align}\label{DB2}
{\rm DB}_t(C_T; T) 	&=-1_{\{\tau>t\}}\bigg\{\mathcal{B}(t,T)\mathbb{E}^{\mathbb{Q}^T}\left[\left.C_T\right|\mathcal{G}^Z_t\right]g(t, T)-\mathbb{E}^{\mathbb{Q}^t}\left[\left.C_t\right|\mathcal{G}^Z_t\right] \\
&-\int_{t}^{T}\frac{\partial\left(\mathcal{B}(t,u)\mathbb{E}^{\mathbb{Q}^u}\left[\left.C_u\right|\mathcal{G}^Z_t\right]\right)}{\partial u}g(t, u)du\bigg\}.\notag
\end{align}
In this way, as the interest rate model follows the Markovian affine model, the partial derivative term in \eqref{DB2} admits a closed-form solution in many cases and we get rid of the need to compute a $T$-partial derivative of $g(t,T)$, which is rather more complicated.
\subsubsection{Examples of concrete insurance contracts}
These formulas for survival and death benefits may still be considered abstract, so we apply them to some concrete insurance or pension products. 

{\bf Longevity Bond}: Consider a unit zero-coupon longevity bond which pays \$1 times $e^{-\int_{t}^{T}\mu_sds}$, the percentage of survivors in a population during $t$ to $T$. \cite{BCD} show that the longevity bond takes the form
$$\mathcal{B}_L(t,T) = \mathbb{E}^{\mathbb{Q}}\left[\left.e^{-\int_{t}^{T}r_s+\mu_sds}\right|\mathcal{F}_t\right].$$
Under the Volterra mortality model with LRD, Proposition \ref{pro1} immediately implies that
\[\mathcal{B}_L(t,T) = \mathcal{B}(t,T)g(t,T),\]
by setting $C_T \equiv 1$ once the financial market is independent of human mortality.

{\bf Annuity}: Consider a $t'$-years deferred annuity involving a continuous payment of an indexed benefit from time t onwards, conditional on survival of the policyholder at that time. Suppose that the payoff is made of a unit amount each year. Denote  $x^*$ as the maximum age humans can live. The fair value of such an annuity is given by
\begin{align}\label{annuity}
\begin{split}
{\rm AN}_t(t') &= \sum_{h=t'}^{x^*-t-1}{\rm SB}_t(1;t+h)
=\sum_{T=t+t'}^{x^*-1} \mathcal{B}(t,T)g(t,T)\\
&=\sum_{T=t+t'}^{x^*-1}e^{\widetilde{\alpha}(t,T)+\widetilde{\beta}(t,T)Z_t}e^{-\int_{0}^{T}m(s)ds+\int_0^t\mu_sds}\exp(Y_t(T)),
\end{split}
\end{align}
where  $Y_t(T)$ is defined in \eqref{Y} and $\widetilde{\alpha}(t,T)$ and $\widetilde{\beta}(t,T)$ are as in \eqref{affineMarkov}.

{\bf Assurances:} Consider an assurance guaranteeing a unit amount benefit in case of death in the period $(t,T]$.  By setting $C\equiv 1$ in \eqref{DB2}, the fair value of such an assurance is given by
\[{\rm AS}_t(T) = 1-\mathcal{B}(t,T)g(t, T)
+ \int_{t}^{T}\frac{\partial\mathcal{B}(t,u)}{\partial u}g(t, u)du, \]
where $\mathcal{B}(t,T)$ is defined in \eqref{affineMarkov} and $g(t,T)$ in Theorem \ref{theorem1}.

{\bf Endowment:} Consider an endowment given the survival on time $t$ with maturity time $T$, which  includes a survival benefit $C_1$ given the survival on time $T$ and a death benefit $C_2$ in case of the death in the period $(t,T]$. $C_1$ and $C_2$ are constants. By  Propositions \ref{pro1} and \ref{pro2} and \eqref{DB2}, the  fair value of such an endowment is given by
\begin{align}
{\rm EN}_t^T(C_1,C_2) &= {\rm SB}_t(C_1; T) + {\rm DB}_t(C_2; T)\notag\\
& = (C_1-C_2)\mathcal{B}(t,T)g(t,T)+C_2\left(1
+ \int_{t}^{T}\frac{\partial\mathcal{B}(t,u)}{\partial u}g(t, u)du\right), \notag
\end{align}
where $\mathcal{B}(t,T)$ is defined in \eqref{affineMarkov} and $g(t,T)$ in Theorem \ref{theorem1}.

\subsection{Esscher transform}\label{SectionEsscher}
Although Propositions \ref{pro1} and \ref{pro2} facilitate the model development under the pricing measure and the calibration to market prices of insurance products, an insurance practice may not have sufficient market prices for such calibration. In addition, risk management requires the connection between the physical and pricing measures as demonstrated in the next section. Therefore, we present two possible ways to link the measures of $\mathbb{P}$ and $\mathbb{Q}$ with limited observed prices. For the time being, we focus on the situation in which the Volterra mortality model is estimated using a historical mortality table and hence built under the physical measure $\mathbb{P} \neq \mathbb{Q}$. 

The first approach commonly used to identify a pricing measure in the actuarial literature is the Esscher transform. \cite{CB} apply the Esscher transform to the mortality rate to find a related martingale measure for pricing longevity derivatives.  \cite{WZJH} also use the Esscher transform for pricing longevity derivatives based on an improved Lee--Carter model.  Although the mortality rate $\mu_t$ is non-Markovian and non-semimartingale under our framework,  the advantage is that we have an explicit Laplace-Fourier functional representation in Theorem \ref{theorem1}. For a random variable $\gamma$ with a well-defined moment-generating function (MGF) under $\mathbb{P}$, an equivalent probability measure $\mathbb{Q}(\theta)$ derived from the Esscher transform with parameter $\theta$ is defined as
\begin{eqnarray}
\frac{d\mathbb{Q}(\theta)}{d \mathbb{P}} = \frac{e^{\theta\gamma}}{\mathbb{E}[e^{\theta\gamma}]}.
\end{eqnarray}

By setting $c_0=1$ and $c_1 = 0$ in Theorem \ref{theorem1}, the MGF for the random variable $-\int_t^T\mu_s\,ds$ is well-defined and can be obtained in an explicit form. Specifically, as we assume $\mu_t = m(t) + \eta X_t$, the MGF defined as 
$$M(\theta_T)= \mathbb{E}[e^{-\theta_T\int_t^T\mu_s\,ds}],$$
which corresponds to the $g(t,T)$ in Theorem \ref{theorem1} with the parameters $m(t)$ and $\eta$ replaced with $\theta_Tm(t)$ and $\theta_T\eta$ for the constant $\theta_T$ and a fixed $T$. For instance, we observe a risk-free zero coupon bond and a zero coupon longevity bond with the same maturity. Then, we can deduce the synthetic value of 
\begin{equation}
\mathbb{E}^{\mathbb{Q}(\theta_T)}_t[e^{-\int_t^T\mu_s\,ds}] = \frac{\mathbb{E}_t[e^{-(\theta_T+1)\int_t^T\mu_s\,ds}]}{\mathbb{E}_t[e^{-\theta_T\int_t^T\mu_s\,ds}]}=\frac{M(\theta_T + 1)}{M(\theta_T)}. \label{Esscher}
\end{equation}
 Although the left-hand quantity is deduced from market prices, the $M(\theta_T)$ achieves a closed-form solution from our model through Theorem \ref{theorem1}. Specifically, $M(\theta)$ is the $g(t,T)$ in Theorem \ref{theorem1} with $m(t)$ and $\eta$ replaced with $\theta m(t)$ and $\theta\eta$, respectively. One can then calibrate $\theta_T$ to the term structure of longevity bonds, or longevity bond prices for different maturity $T$, after estimating the physical model parameters, including the LRD feature, using historical data. 

From \eqref{Esscher}, when $\theta_T = 0$, the longevity bond is priced under $\mathbb{P}$ and our previous valuation formulas hold. For a nonzero $\theta_T$, a slight adjustment can be made through \eqref{Esscher} as the MGF is explicitly known.

\subsection{Affine retaining transform}\label{retaining}
Although the Esscher transform provides us with a powerful and convenient framework to identify a pricing measure, it does not offer us an explicit stochastic process under the pricing measure. When we perform a risk management strategy, we need the stochastic process of the mortality rate under both $\mathbb{P}$ and $\mathbb{Q}$. It is desirable that the Volterra mortality model retains the affine nature in Definition \ref{def1}. Therefore, we propose the following affine retaining transform based on the Girsanov theorem.

\begin{definition}\label{def2}
Given an affine SIVE of \eqref{volterra process} satisfying Definition \ref{def1}, an affine retaining transform for measure change is based on shifting the Wiener process as follows:
$$dW^\mathbb{Q}_t = dW_t - \sigma(X_t)^\top\varphi(t)\,dt,$$
for a deterministic function $\varphi(t)\in \mathbb{R}^d$ satisfying 
\begin{equation}\label{condition}
\mathbb{E}_t\left[e^{\frac{1}{2}\int_0^T|\sigma(X_t)^\top\varphi(t)|^2dt}\right]< \infty.
\end{equation}
\end{definition}

Under Definition \ref{def2}, we identify a pricing measure $\mathbb{Q}$ equivalent to $\mathbb{P}$:
$$\frac{d\mathbb{Q}}{d \mathbb{P}} = e^{-\frac{1}{2}\int_0^t|\sigma(X_s)^\top\varphi(s)|^2ds+ \int_0^t\varphi(s)^\top\sigma(X_s)dW_s},$$
where $\varphi(t)$ is calibrated to observed prices. In addition, the mortality process $\mu_t = m(t) +\eta X_t$ in \eqref{volterra process} under $\mathbb{Q}$ has the $X_t$ changed to
\begin{eqnarray}\label{XunderQ}
  X_t = X_0 +\int_0^tK(t-s)(b(X_s)+a(X_s)\varphi(s))ds +\int_0^t K(t-s)\sigma(X_s)dW_s^\mathbb{Q},
\end{eqnarray}
where $b(X_s)+a(X_s)\varphi(s)$ and $a(X_s)$ still satisfy the affine nature in Definition \ref{def1}. Hence, the pricing formulas of Propositions \ref{pro1} and \ref{pro2} remain the same except that the $b(X_s)$ is replaced with $b(X_s)+a(X_s)\varphi(s)$ once the affine retaining transform in Definition \ref{def2} is adopted. 

\begin{remark} Although the Esscher and affine retaining transforms presented in Sections \ref{SectionEsscher} and \ref{retaining} are applied to the Volterra mortality model, these techniques have been widely used in the actuarial science literature, including the measure change with the affine interest rate models. Therefore, we do not repeat the detailed case for the interest rate. We mention them to highlight the advantage of the proposed LRD mortality model in sense of calibrating to the pricing measure.
\end{remark}

\section{Optimal hedging of longevity risk}\label{Hedge}
We further investigate optimal hedging with the proposed LRD mortality model, as hedging is a typical risk management task. The intent is to demonstrate the tractability of the LRD mortality model in hedging problems. As hedging should be performed under the physical probability measure $\mathbb{P}$, whereas longevity securities such as the longevity bonds and swaps are valued in the market-implied pricing measure $\mathbb{Q}$, we adopt the affine retaining transform detailed in Section \ref{retaining} to bridge the two probability measures in this section.

Let us sketch the conceptual framework prior to detailing the mathematics. As insurance product prices under the Volterra mortality model are semimartingales and hence can be expressed in SDE, the insurer's wealth also satisfies a SDE with stochastic coefficients, which are possibly non-Markovian. According to stochastic control theory, the insurer's wealth plays the role of the state process. Therefore, the theory of backward SDE (BSDE) is useful for solving the stochastic optimal control problem for a state process with stochastic coefficients. Typically, the mean-variance (MV) hedging problem is closely related to the linear-quadratic (LQ) control problem under the classic formulation of the BSDE approach.  In the following, we leverage this well-received theoretical result to show the application of the LRD mortality model, though the optimal hedging derived is novel and has remarkable performance in reducing risk with the LRD mortality. The performance is, however, shown in the next section numerically.

\subsection{Problem formulation}
Consider an insurer offering a pension scheme who wants to hedge the longevity risk using a longevity security. Specifically, the insurer allocates her capital among a bank account, risk-free zero-coupon bond, and zero-coupon longevity bond. Let us concentrate on the one-dimensional case so that $d=k=1$ from now on.

To simplify the discussion, we adopt the VV mortality rate and assume $m(t)=0$ and $\eta=1$ in \eqref{mu0}. In other words, $\mu(t)= X(t)$ and
\begin{equation}\label{mortality}
\mu_t= X_t = X_0+\int_{0}^{t}K(t-s)(b^0-b^1X_s)ds+\int_{0}^{t}K(t-s)\sigma_{\mu}dW_{s},
\end{equation} 
where $b^0$, $b^1$, and $\sigma_{\mu}$ are constants and $K$ is the Volterra kernel.  In addition, the interest rate $r_t = Z_t$  follows the Vasicek model: 
\begin{equation}
dr(t)=(\widetilde{b}^0 - \widetilde{b}^1r_t)dt +\sigma_rdW'_{t},
\end{equation}
where $\widetilde{b}^0$, $\widetilde{b}^1$, and $\sigma_r$ are constant parameters.   $W_t$ and $W'_t$ are independent Wiener processes under $\mathbb{P}$. Let $\bm{W}(t) = (W_t, W'_t)^\top$. Using the affine retaining transform in Definition \ref{def2}, the Weiner process under the pricing measure is given by
 \[dW_t^{\mathbb{Q}} = dW_t - \sigma_{\mu}\frac{\varphi(t)}{\sigma_{\mu}}dt, ~ d{W'_t}^{\mathbb{Q}} = dW'_t - \sigma_r\frac{\vartheta(t)}{\sigma_r}dt,\]
 where $\vartheta$ and $\varphi$ are deterministic functions satisfying the condition \eqref{condition}.  Under the pricing measure, the mortality and interest rates are, respectively,
 \[X_t = X_0+\int_{0}^{t}K(t-s)(b^0+ \varphi(s)\sigma_{\mu}-b^1X_s)ds+\int_{0}^{t}K(t-s)\sigma_{\mu}dW_s^{\mathbb{Q}};\]
$$dr(t)=(\widetilde{b}^0 + \vartheta(t)\sigma_r- \widetilde{b}^1r_t)dt +\sigma_rd{W'_t}^{\mathbb{Q}}.$$  

As the unit zero coupon bond price takes the form
\[\mathcal{B}(t,T)=\mathbb{E}^{\mathbb{Q}}\left[\left.e^{-\int_{t}^{T}r(s)ds}\right|\mathcal{F}_t\right] = e^{\widetilde{\alpha}(t, T)+\widetilde{\beta}(t, T)r_t}, \]
with $\widetilde{\alpha}(t,T)$ and $\widetilde{\beta}(t,T)$  as defined in Appendix \ref{appendix:affine}, the $\mathbb{P}$-dynamics of the bond reads
\[d\mathcal{B}(t,T) = \mathcal{B}(t,T)(r(t)+\nu_{\mathcal{B}}(t))dt + \mathcal{B}(t,T)\sigma_b(t)dW'_t,\]
where $\nu_{\mathcal{B}} = \vartheta(t)\sigma_b(t)$ and $\sigma_b(t) = -\widetilde{\beta}(t,T)\sigma_r$.
Similarly, using the expression for a zero coupon longevity bond, i.e., 
\begin{align}
\mathcal{B}_L(t,T) =\mathbb{E}^{\mathbb{Q}}\left[\left.e^{-\int_{t}^{T}r(s)+\mu(s)ds}\right|\mathcal{F}_t\right] =\mathcal{B}(t,T)e^{\int_{0}^{t}\mu(s)ds}{\rm exp}(Y^1_t(T)),\notag
\end{align}
where $Y^1_t(T)$ is equivalent to the $Y_t(T)$ in \eqref{Y} with $b(x) = b^0+\varphi(s)\sigma_{\mu} - b^1x$,  $\sigma(x) = \sigma_{\mu}$, and $W$ replaced by $W^{\mathbb{Q}}$, we obtain the $\mathbb{P}$-dynamics of the longevity bond prices as follows:
\[d\mathcal{B}_L(t,T) = \mathcal{B}_L(t,T)(r(t)+\mu(t)+\nu_L(t))dt + \mathcal{B}_L(t,T)\sigma_l(t) dW_t +\mathcal{B}_L(t,T)\sigma_bdW'_t,\]
where $\nu_L = \nu_{\mathcal{B}} + \varphi(t)\sigma_l$, $\sigma_l = -\psi_1(T-t)\sigma_{\mu}$, and $\psi_1\in\mathcal{L}^2([0,T],\mathbb{C})$ is the solution of the Riccati equation $\psi_1 = (-1- b^1\psi_1)*K$.   As an investment amount of $\mathcal{B}_L(t,T)$ in the longevity bond at time $t$ becomes $e^{-\int_{t}^{\tau}\mu(s)ds}\mathcal{B}_L(\tau,T)$ at $\tau >t$,  the value of holding one unit of zero coupon longevity bond $\mathcal{B}_L(t)$ satisfies
\begin{equation}\label{longevity}
d\mathcal{B}_L(t,T) = \mathcal{B}_L(t,T)(r(t)+\nu_L(t))dt + \mathcal{B}_L(t,T)\sigma_l(t) dW_t +\mathcal{B}_L(t,T)\sigma_bdW'_t.
\end{equation}
The quantities $\nu_L - \nu_{\mathcal{B}}$ and $\nu_{\mathcal{B}}$  are often known as the market prices of mortality and interest rate risks, respectively. From \eqref{longevity}, the zero coupon longevity bond price still satisfies a SDE due to the semimartingale nature of $Y_t(T)$. This fact enables us to deal with the optimal hedging problem with a LRD mortality rate. Note that the LRD feature is reflected by the volatility term of $\mathcal{B}_L(t)$ through a Riccati-Volterra equation. 

Let $u_0(t)$, $u_1(t)$, and $u_2(t)$ denote the investment amounts in the bank account, zero-coupon longevity bond, and zero-coupon bond respectively. Denote $\tilde{N}(t)$ as a stochastic Poisson process with intensity $k_1\mu(t)$ and $\{z_i\}_{i=1}^{\infty}$ as independent identically distributed (iid) insurance claims. Consider a hedging horizon of $T_0<T$. Then, the wealth process of the insurer reads
\begin{equation}
M(t)=u_0(t)+u_1(t)+u_2(t)- \sum_{i=1}^{\tilde{N}(t)}z_i-\Pi(t),~ t\in[0,T_0],
\end{equation}
where $\Pi = \int_{0}^{t}\pi(s)ds$, $t\in[0,T_0]$, and $\pi(t)$ is a $\mathcal{F}_t$-adapted, square integrable process representing the pension annuity net cash outflow. We denote the filtration generated by $\{M(s):0\leq s\leq t\}$ by $\tilde{\mathcal{H}}_t \supseteq \mathcal{F}_t$. The insurer's wealth $M(t)$ satisfies the following SDE:
\begin{equation}\label{wealth}
dM(t)= (M(t)r(t)+u(t)^\top\nu(t)-\pi(t))dt + u(t)^\top\sigma_S(t)^\top d{\bm W}(t)-zd\tilde{N}(t),
\end{equation}
where $z$ has the same distribution as $z_1$, $u(t)= (u_1(t),u_2(t))^\top$, $\nu(t)=(\nu_L(t),\nu_{\mathcal{B}}(t))^\top$, and
\[\sigma_S(t)^\top=\begin{pmatrix}
\sigma_l&\sigma_b\\
0&\sigma_b
\end{pmatrix}.\]

If a hedging strategy $u(t)$ is a $\mathcal{F}_t$-adapted process and  $\mathbb{E}[\int_{0}^{T_0}|u(s)|^2ds]<\infty$, then it is said to be admissible.  We denote the set of admissible controls as $\mathcal{U}$. 
\begin{definition} \label{def3} The classic mean-variance (MV) hedging problem is defined as
\begin{equation}\label{min}
V(\phi) = \mathop{{\rm min}}\limits_{u(\cdot)\in\mathcal{U}} {\rm Var}(M(T_0))-\frac{\phi}{2}\mathbb{E}[M(T_0)],
\end{equation}
where the parameter $\phi$ measures the insurer's risk averseness. 
\end{definition}

When $\phi=0$, problem \eqref{min} refers to the minimum-variance hedging. For any given $\bar{M}=\mathbb{E}[M(T_0)]$,
\[\mathbb{E}[(M(T_0) - \bar{M})^2]- \frac{\phi}{2}\mathbb{E}[M(T_0)] = \mathbb{E}[(M(T_0)-(\bar{M}+\frac{\phi}{4}))^2]-\frac{\phi}{2}\bar{M} - \frac{\phi^2}{16}.\]
In addition, the MV hedging problem can be embedded into a target-based objective. Specifically, the problem \eqref{min} is equivalent to
\begin{equation}\label{dmin}
\mathop{{\rm min}}\limits_{\bar{M}\in \mathbb{R}} \mathop{{\rm min}}\limits_{u(\cdot)\in\mathcal{U}}\mathbb{E}[(M(T_0)-c)^2]-\frac{\phi}{2}\bar{M} - \frac{\phi^2}{16},
\end{equation}
where $ c= \bar{M}+\frac{\phi}{4}$. The inner minimization problem there refers to a target-based objective that aims to make the wealth close to the target $c$.

\subsection{Hedging mortality with LRD}
Let $\pi(t) = k_2e^{-\int_{0}^{t}\mu(s)ds}$ and $\Sigma(t) = \sigma_S(t)^\top\sigma_S(t)$.  To solve the optimal hedging problem, we introduce two additional probability measures:
\begin{align}\label{new measure}
\frac{d\hat{\mathbb{P}}}{d\mathbb{P}}= e^{-\int_{0}^{t}\xi(s)^\top d\bm{W}(s)-\frac{1}{2}|\xi(s)|^2ds}, ~\frac{d\acute{\mathbb{P}}}{d\mathbb{P}}= e^{-\int_{0}^{t}\zeta(s)^\top d\bm{W}(s)-\frac{1}{2}\zeta(s)^\top\zeta(s)ds}\notag
\end{align}
 with $\xi(t) =(2\varphi(t), 2\vartheta(t))^\top$ and  $\zeta(t) = (\varphi(t),\vartheta(t))^\top$. By the Girsanov theorem, $\hat{\bm{W}}_t \triangleq {\bm W}_t + \int_{0}^{t}\xi(s)ds$ and  $\acute{\bm{W}}_t \triangleq {\bm W}_t + \int_{0}^{t}\zeta(s)ds$ are Wiener processes under $\hat{\mathbb{P}}$ and $\acute{\mathbb{P}}$, respectively. 
Denote $\hat{\mathbb{E}}[\cdot]$ and $\acute{\mathbb{E}}[\cdot]$ as expectations under $\hat{\mathbb{P}}$ and $\acute{\mathbb{P}}$, respectively.  By Theorem \ref{theorem1}, 
\begin{equation}
\acute{\mathbb{E}}\left[\left.e^{-\int_{0}^{s}\mu_{\tau}d\tau}\right|\tilde{\mathcal{H}}_t\right] = {\rm exp}(Y_t^2(T)),\nonumber
\end{equation}
where $Y^2_t(T)$ is equivalent to the $Y_t(T)$ in \eqref{Y} with $b(x) = b^0-\varphi(s)\sigma_{\mu}-b^1x$,  $\sigma(x) = \sigma_{\mu}$, and $W$ replaced by $\acute{W}$; $\acute{\mathbb{E}}[\mu_s|\tilde{\mathcal{H}}_t] = \acute{\mathbb{E}}[X_s|\tilde{\mathcal{H}}_t]$ is equivalent to $\mathbb{E}[X_s|\mathcal{F}_t]$ as defined in \eqref{expectation} with $B = -b^1$,  $b^0(s)$ replaced by $b^0-\varphi(s)\sigma_{\mu}$,  and $W$ replaced by $\acute{W}$. 
In addition, we have the following expressions.
\begin{align}
\hat{\mathbb{E}}\left[\left.e^{-2\int_{t}^{T_0}r(s)ds}\right|\mathcal{F}_t\right] &= {\rm exp}(\alpha_1(t,T_0)+ \beta_1(t, T_0)r(t)),\label{exp2r}\\
\acute{\mathcal{B}}(t,s)=\acute{\mathbb{E}}\left[\left.e^{-\int_{t}^{s}r(u)du}\right|\mathcal{F}_t\right] &={\rm exp}(\alpha_2(t,s)+ \beta_2(t,s)r(t)),\label{expr}
\end{align}
where $\alpha_1(t,T_0)$, $\beta_2(t, T_0)$, $\alpha_2(t,s)$, and $\beta_2(t, s)$ solve the ODEs in Appendix \ref{appendix:affine}.  The following theorem provides the optimal hedging strategy. 
\begin{theorem}\label{theorem 2}
 Consider two stochastic processes
\begin{equation}\label{P}
P(t) =  \frac{e^{-\int_{t}^{T_0}\vartheta^2(s)+ \varphi^2(s)ds}}{\hat{\mathbb{E}}\left[\left.e^{-2\int_{t}^{T_0}r(s)ds}\right|\mathcal{F}_t\right]}
\end{equation}
and 
\begin{equation}\label{Q}
Q(t) = -P(t)[Q_0(t)+c\acute{\mathcal{B}}(t,T_0)],
\end{equation}
where 
\[Q_0(t) =\int_{t}^{T_0}\acute{\mathcal{B}}(t,s)(k_1\mathbb{E}[z]\acute{\mathbb{E}}[\mu_s|\tilde{\mathcal{H}}_t]+ k_2\acute{\mathbb{E}}[e^{-\int_{0}^{s}\mu_{\tau}d\tau}|\tilde{\mathcal{H}}_t])ds,\]
\[\acute{\mathcal{B}}(t,s)=\acute{\mathbb{E}}\left[\left.e^{-\int_{t}^{s}r(u)du}\right|\mathcal{F}_t\right],~ 0\leq t\leq s.\]
Once 
\begin{equation}
dP(t) = \mu_P(t)dt + \eta_1^\top d{\bm W}(t)~  \hbox{and} \quad dQ(t) = \mu_Q(t)dt + \eta_2^\top d{\bm W}(t) \label{PQ}
\end{equation}
under $\mathbb{P}$, the inner minimization problem in \eqref{dmin} has an optimal feedback control: 
\begin{equation}\label{ustar}
u^*_c(t) = -\Sigma(t)^{-1}\left[\left(\nu(t) + \frac{\sigma_S(t)^\top\eta_1(t)}{P(t)}\right)M(t)+ \frac{Q(t)\nu(t)+ \sigma_S(t)^\top\eta_2(t)}{P(t)}\right].
\end{equation}
In addition, the optimal objective value is $P(0)(M(0)+ \frac{Q(0)}{P(0)})^2 + I(0)$, where 
\begin{equation}\label{I(t)}
I(t) = \mathbb{E}\left[\left.\int_{t}^{T_0}P\left\{\mu z^2 + \left(\frac{\eta_2- Q\eta_1}{P^2}\right)^\top\sigma^\perp\left(\frac{\eta_2- Q\eta_1}{P^2}\right)\right\}(s)ds\right|\tilde{\mathcal{H}}_t\right]
\end{equation}
in which $\sigma^\perp = {\rm id} - \sigma_S(t)\Sigma(t)^{-1}\sigma_S(t)^\top$. 
\end{theorem}
\begin{proof}
 See Appendix \ref{appendix:proof}.
\end{proof}
\begin{pro}\label{proeta}
Then, the diffusion coefficients in \eqref{PQ} are $\eta_1 = (0, \eta_{12})^\top$, where $\eta_{12} = -P(t)\beta_1(t, T_0)\sigma_r$ and $\eta_2 = (\eta_{21}, \eta_{22})^\top$ in which
	\[\eta_{21} = -P(t)\int_{t}^{T_0}\acute{\mathcal{B}}(t,s)\left(k_1\mathbb{E}[z]E_{B}(s-t)\sigma_{\mu}+ k_2\acute{\mathbb{E}}\left[\left.e^{-\int_{0}^{s}\mu_{\tau}d\tau}\right|\tilde{\mathcal{H}}_t\right]\psi_2(s-t)\sigma_{\mu}\right)ds,\]
	\begin{align}
	\eta_{22} &= -P(t)\bigg\{\int_{t}^{T_0}\acute{\mathcal{B}}(t,s)\left(k_1\mathbb{E}[z]\acute{\mathbb{E}}[\mu_s|\tilde{\mathcal{H}}_t]+ k_2\acute{\mathbb{E}}\left[\left.e^{-\int_{0}^{s}\mu_{\tau}d\tau}\right|\tilde{\mathcal{H}}_t\right]\right)\beta_2(t, s)\sigma_rds \notag\\
	&+ c\acute{\mathcal{B}}(t, T_0)\beta_2(t, T_0)\sigma_r\bigg\}
	+ P(t)[Q_0(t)+c\acute{\mathcal{B}}(t, T_0)]\beta_1(t, T_0)\sigma_r,
	\end{align}
where $\beta_1(t, T_0)$ is defined in \eqref{exp2r}, $\beta_2(t, s)$ in \eqref{expr}, $E_{B}$ in Theorem \ref{theorem1} with $B = -b^1$, and $\psi_2\in \mathcal{L}^2([0,s],\mathbb{C})$ solves the Riccati equation $\psi_2 = (-1 - \psi_2b^1)*K$. 
\end{pro}

\begin{pro}\label{hedgingstrategy}
	The optimal  hedging strategy $u^*(t) =(u_1^*(t),u_2^*(t))^\top$  to problem  \eqref{min} is given by
	\begin{align}
	u_1^*(t)=&-\frac{1}{\sigma_l(t)}\left\{\left[M(t)-Q_0(t)-\left(\bar{M}^* + \frac{\phi}{4}\right)\acute{\mathcal{B}}(t,T_0)\right]\varphi(t) + \frac{\eta_{21}(t)}{P(t)}\right\}\label{u1},\\
	u_2^*(t)=&-\frac{1}{\sigma_b(t)}\left\{\left[M(t)-Q_0(t)-\left(\bar{M}^* + \frac{\phi}{4}\right)\acute{\mathcal{B}}(t,T_0)\right]\vartheta(t) + \frac{M(t)\eta_{12}(t)+\eta_{22}(t)}{P(t)}\right\}\notag\label{u2}\\
	&-u_1^*(t),
	\end{align}
	where 
	\[\bar{M}^* = \frac{\frac{\phi}{4}(1-P(0)\acute{\mathcal{B}}^2(0,T_0))+P(0)\acute{\mathcal{B}}(0,T_0)(M(0)-Q_0(0))}{P(0)\acute{\mathcal{B}}^2(0,T_0)}.\]
	\label{pro:u}
\end{pro}

The explicit optimal hedging strategy  in Proposition \ref{hedgingstrategy} incorporates the LRD feature through $\eta$ which depends on the mortality rate path and the kernel $K$ as shown in Proposition \ref{proeta}. In addition, the Hurst parameter is contained in the  kernel function $K$. 

\section{Impact of LRD: Numerical studies}\label{numerical}
In this section, we numerically examine the impact of long-range dependence on the prices of insurance products and the hedging effectiveness. To do so, we contrast the LRD mortality model with its Markovian counterpart. For the latter case, the Hurst parameter $H$ is set to 1/2. As the LRD appears when $H >1/2$, we examine the effect when $H$ falls into this range.

\subsection{Survival probability}
\label{section:survival}
As the basic quantity, we begin with the survival probability. Under the Volterra mortality model, we assume that process $X$ satisfies Equation \eqref{fractional X} which is a Volterra type of Vasicek model.  The Vasicek model is a special case with $\alpha =1$ or $H=1/2$. We compare the Vasicek and VV mortality models using two different values of $H$ while the other parameters are kept constant. It is empirically estimated by \cite{YPC} that the $H$ is around 0.83 for mortality data. Thus, we choose an $\alpha$ of 1.33 for the VV mortality model.  Table \ref{parameters} summarizes the remaining parameters used in this numerical study. \textcolor{black}{ The parameters chosen have similar magnitudes to those in \cite{Biffis} for the case of Markovian model.}
     
\begin{table}[H]
	\centering                                                                                                     
	\begin{tabular}{lcccccccc}  
		\toprule
		\toprule
		Projection & $\alpha$ &\textcolor{black}{$m(t)$} & $\eta$ & $\lambda$ & $\theta$ & $\sigma$ & $t$ & $X_0$\\
		\midrule
		\midrule
		A&1.33& SIM92 & 0.2 & 0.5 & 0.0009 & 0.01&40&0.001\\
		B&1& SIM92 & 0.2 & 0.5 & 0.0009 & 0.01&40&0.001\\
		\bottomrule
	\end{tabular}
	\caption{Parameter values for the mortality model}
	\label{parameters}
\end{table}
\textcolor{black}{\begin{remark}\label{remark:paramater}
	The SIM92 in Table \ref{parameters} is a dataset from the Italian National Institute of Statistics (ISTAT) which reports Italian population life tables.  SIM92 is usually employed to price assurance.  Such a setting for $m(t)$ has been adopted in \cite{Biffis}. Specifically, after fixing the other parameter values, the $m(t)$ is calibrated to fit the SIM92 table, so the functional form of $m(t)$ is not explicitly shown here.
\end{remark}}

\textcolor{black}{Although parameter values are assigned in this numerical experiment, we stress that, in reality, the parameters can be calibrated to observed prices of actuarial products using the set of the closed-form pricing formulas derived in this paper. In addition, the parameter $\theta$ in \eqref{survival} or $b^0$ in Definition \ref{def1} can be set as a bounded measurable function of time $t$ rather than a constant as in our example.}

 \textcolor{black}{In Table \ref{parameters}, the symbol $t$ stands for the age group. For example, when we set $t = 40$, it corresponds to a group of the survival population at the age of 40. In Figures \ref{fig1}(a) and \ref{fig2}(a), we simulate two different sample paths of  $X$ for this group of individuals over the time interval $[0,t]$.}  
  Under the VV mortality model, the historical sample paths of $X$ affect the estimated survival probability, whereas the Vasicek model does not due to its Markovian nature.  Given the parameters in Table \ref{parameters} and  \eqref{survival}, we directly calculate survival probabilities from the two models.   	\textcolor{black}{ By \eqref{survival} and Theorem \ref{theorem1}, 
  	\begin{align}\label{markov survival}
  	&\mathbb{P}(\tau> T|\mathcal{F}_t) \notag \\
  	&= e^{-\int_{t}^{T}m(s)ds}\exp\left( -\eta\int_{t}^{T}\mathbb{E}[X_s|\mathcal{F}_t]ds+ \frac{1}{2}\int_{t}^{T}\psi(T-s)a(\mathbb{E}[X_s|\mathcal{F}_t])\psi(T-s)^\top ds\right), 
  	\end{align}
 for $T > t$.  Under the Vasicek mortality model, the survival probability depends only on $X_t$ $(t = 40)$ as $\mathbb{E}[\mu_s|\mathcal{F}_t] = \mu_t$. However, under the VV mortality model, the expression of $\mathbb{E}[\mu_s|\mathcal{F}_t]$ given in \eqref{expectation} depends on the whole historical path of $X$. Based on the simulated sample paths, we calculate the survival probabilities for the interval $T \in[t, x^*]$, where we set the maximum age at $x^* = 109$.} 

Figures \ref{fig1}(b) and \ref{fig2}(b) show the survival probabilities \textcolor{black}{that correspond to the historical records in Figure \ref{fig1}(a) and \ref{fig2}(a), respectively.} The solid line is the survival probability curve with LRD and the dashed line is that of the Markovian model. Depending on the historical record, the LRD survival probability can be higher or lower than the Markovian survival probability.  \textcolor{black}{This indicates that the historical sample path has impact on the survival probability when LRD is present.}  The effect is more pronounced for the middle age group. This is reasonable because the young age group has a shorter historical record and the old age group may be restricted by the human age limit. This kind of middle-age effect may result in a significant effect on insurance pricing. We further examine it with a concrete insurance product.
 
 \begin{figure}[H]
	\centering
	\subfigure[]{\includegraphics[width=7.5cm]{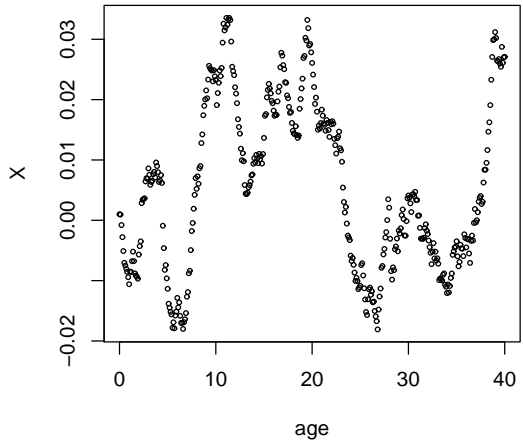}}
	\subfigure[]{\includegraphics[width=7.5cm]{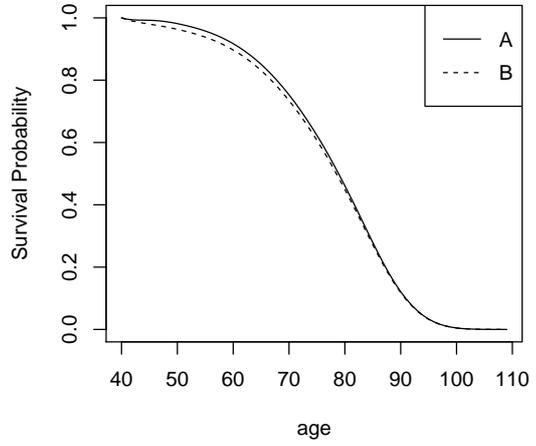}}
	\caption{A sample historical path of $X$ that makes the survival probability with LRD higher than its Markovian counterpart.}
	\label{fig1}
\end{figure}

\begin{figure}[H]
	\centering
	\subfigure[]{\includegraphics[width=7.5cm]{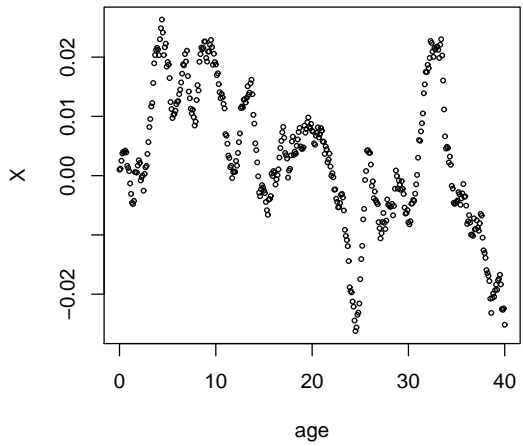}}
	\subfigure[]{\includegraphics[width=7.5cm]{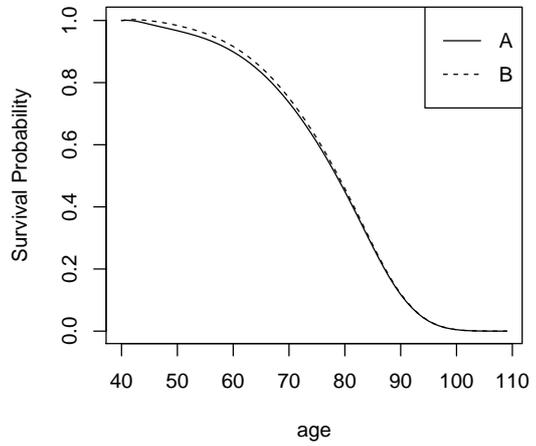}}
	\caption{A sample of historical path of $X$ that makes the survival probability with LRD lower than its Markovian counterpart.}
	\label{fig2}
\end{figure}

\subsection{Impact on annuity}
\textcolor{black}{To examine the effect of LRD on annuity prices, we compare the prices calculated by the two models.} We are interested in annuities because they are popular insurance and pension products around the globe.

The numerical experiment is constructed as follows. Consider a 20-year deferred annuity and its payoff is a unit amount each year. For simplicity, we assume that $\mathbb{Q} = \mathbb{P}$ in this part so that no additional effort is required to identify the pricing measure.  The simulation and calculation are made with the parameters in Table \ref{parameters}. In addition, we specify the short interest rate $r_t = Z_t$ as follows.
$$dZ_t = (\widetilde{b}^0 - \widetilde{b}^1Z_t)dt + \sigma_rdW',$$
where $\widetilde{b}^0=0.01$, $\widetilde{b}^1=0.5$, $\sigma_r=0.3$, and $Z(40)=0.01$. Then we use \eqref{annuity} directly to calculate the price of the annuity and $t' = 20$.
\begin{figure}[H]
	\centering
	\subfigure[]{\includegraphics[width=7.5cm]{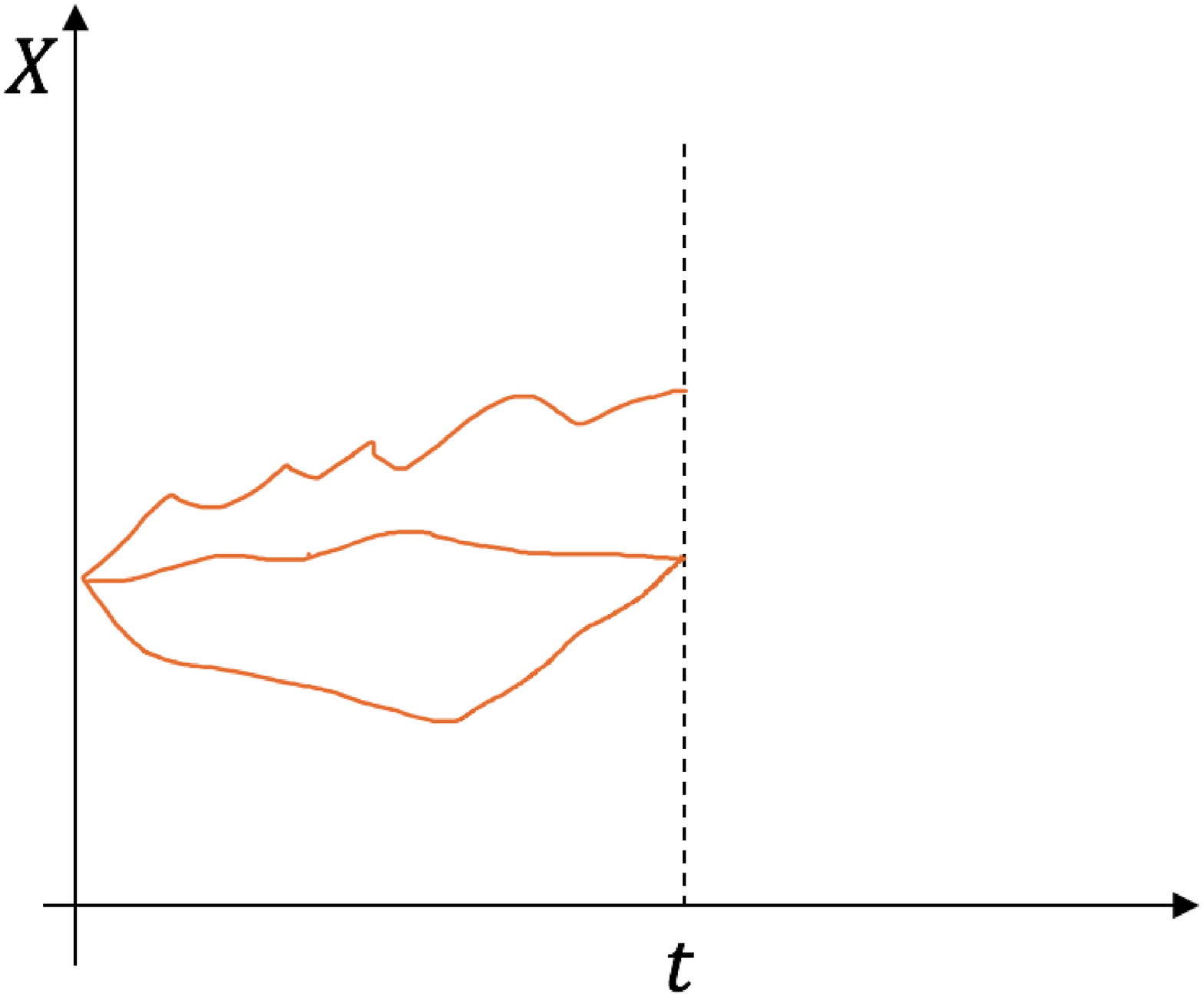}\label{subfig3}}
	\subfigure[]{\includegraphics[width=7.5cm]{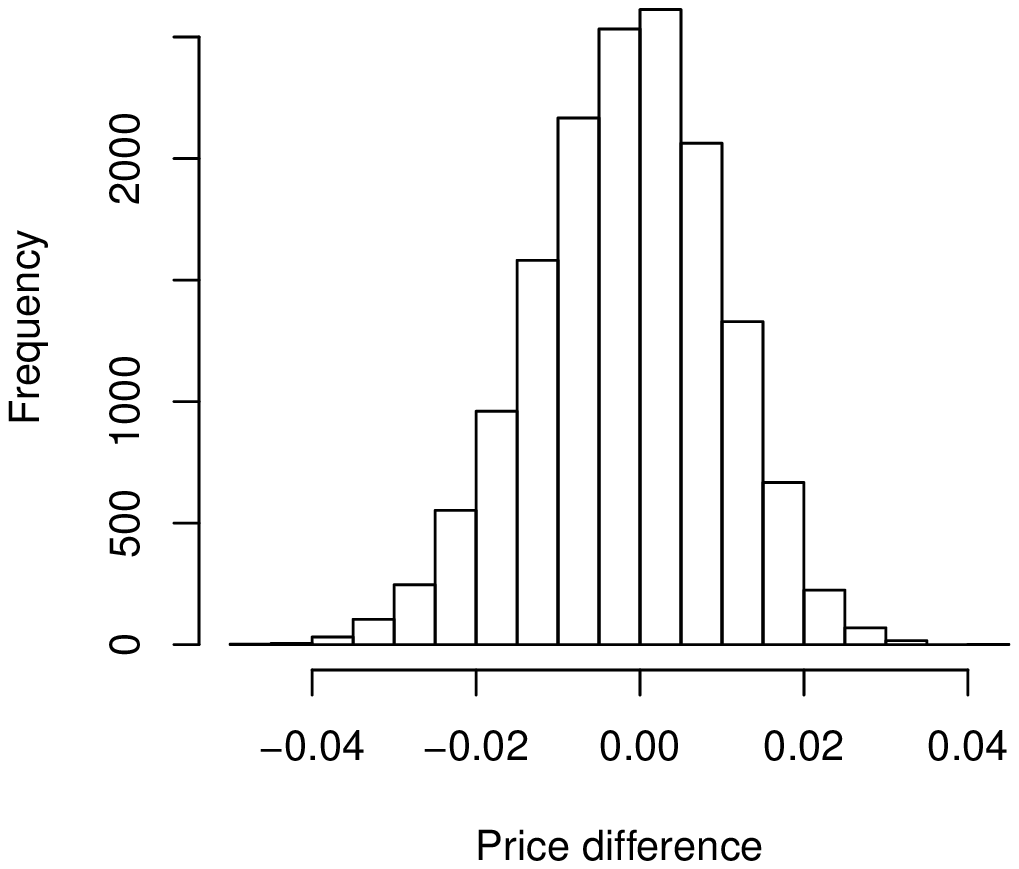}\label{subhist}}
	\caption{\textcolor{black}{ (a) Examples of historical paths for $X$ and (b) histogram of percentage difference in annuity prices between the two models}}
	\label{histogram}
\end{figure}

To demonstrate the LRD effect, we generate 15,000 sample paths of $X$ over the time interval $[0,t]$. In Figure \ref{subfig3}, we illustrate that the last two sample paths meet at time $t$. The classic Markovian model ignores how they come to this point and assigns the same price to the two scenarios \textcolor{black}{as explained in \eqref{markov survival}.}  However, our LRD mortality model takes the historical record into account and assigns two different prices \textcolor{black}{ as shown in \eqref{annuity} and Theorem \ref{theorem1}.}  The problem is to determine how large the difference between these two models is. Clearly, the difference is not a single number as there are uncountably many ways to reach the same point. Therefore, we examine the distribution of the price difference for different historical paths. 

To do so,  Figure \ref{subhist} plots a histogram of the percentage difference of the annuity prices between the LRD and Markovian models. First, the mean of the distribution is near zero, implying that the Markovian mortality model offers an appropriate estimate of the averaged price even under the LRD feature. However, the dispersion of the histogram is still obvious. The price difference between the two models can reach 4\% even for a linear annuity product, and this 4\% difference seems not negligible in practice. The discrepancy may be amplified for products with leveraging effects such as those with optionality. Even for this annuity product, we can see the volatility could be higher compared to the Markovian model due to incorrect predictions of the mortality rate if the realized mortality has the LRD feature.

\textcolor{black}{To illustrate the influence of LRD on products with optionality, consider a European call option on a zero-coupon longevity bond $\mathcal{B}_{L}(t, T)$ with strike $D$ and expiration time $T_1$, where $T$ is the fixed maturity of the bond and $T_1$ is the expiration date of the option so that $0 \leq t \leq T_1< T$. Specifically, the call option payoff reads $V_{0}(\mathcal{B}_{L}(T_1, T)) = \max(\mathcal{B}_{L}(T_1, T) - D, 0)$. We want to focus on the effect of LRD mortality rate, and therefore assume a constant interest rate $r$ and $m(\cdot) =0$. By \eqref{fractional X} and \eqref{longevity}, we have
\begin{eqnarray}
d\mathcal{B}_L(t, T) =  \mathcal{B}_L(t, T)\left[rdt + \psi(T- t)\sigma dW^{\mathbb{Q}}_t\right], 
\label{BLQ}
\end{eqnarray}
under the pricing measure, where $\psi$ solves $\psi = (-\eta -\lambda \psi )*K$. As \eqref{longevity} is the dynamic of $\mathcal{B}_L(t, T)$ under $\mathbb{P}$, the corresponding $\mathbb{Q}$ dynamics in \eqref{BLQ} is one in which the term $\nu_L$ in \eqref{longevity} is absorbed into the $\mathbb{P}$-Brownian motion to form a $\mathbb{Q}$-Brownian motion. Hence, the call value function $V_0(\mathcal{B}_L, t)$ resembles the Black-Scholes formula. Specifically, 
\begin{eqnarray*}
V_0(\mathcal{B}_L, t) &=& \Phi(d_1)\mathcal{B}_L(t, T)-\Phi(d_2)De^{-r(T_1-t)},\cr
d_1 &=& \frac{1}{\psi(T-t)\sigma\sqrt{T_1-t}}\left[\ln\left(\frac{\mathcal{B}_L(t, T)}{D}\right) + \left(r + \frac{1}{2}\psi^2(T-t)\sigma^2\right)(T_1-t) \right], \cr
d_2 &=& d_1 - \psi(T-t)\sigma\sqrt{T_1-t},
\end{eqnarray*}
where $\Phi(\cdot)$ is the cumulative distribution function of the standard normal distribution.}

\textcolor{black}{Let us make a numerical comparison in terms of percentage difference in option price between the VV and Markovian models. Let $r = 0.01$,  $T = 5$, and $T_1 = 2$,  and set the other parameters as in Table \ref{parameters}. Assume $\mathcal{B}_L(t, T) = 0.8 \triangleq D_0$, the benchmarking at-the-money (ATM) strike, at the option issuance time. Note that the historical path of the mortality rate is subsumed into the longevity bond price $\mathcal{B}_L(t, T)$. By varying the strike $D$ from 0.8 (ATM) to 0.832 (4\% in-the-money), option prices under the two models are shown in Figure \ref{fig:option1} while the percentage difference in price is shown in Figure \ref{fig:option2}. When the strike increases by 4\%, the percentage difference in option price could reach 20\% which is quite significant. We mention the 4\% increase in strike because the price of an annuity can reach a 4\% difference in price in the former analysis. When the strike is set to make the option ATM, the difference in the longevity bond price results in a 4\% difference in setting the ATM strike. This example shows that optionality may further amplify the pricing difference. }

\begin{figure}[H]
	\centering
	\subfigure[]{\includegraphics[width=7cm]{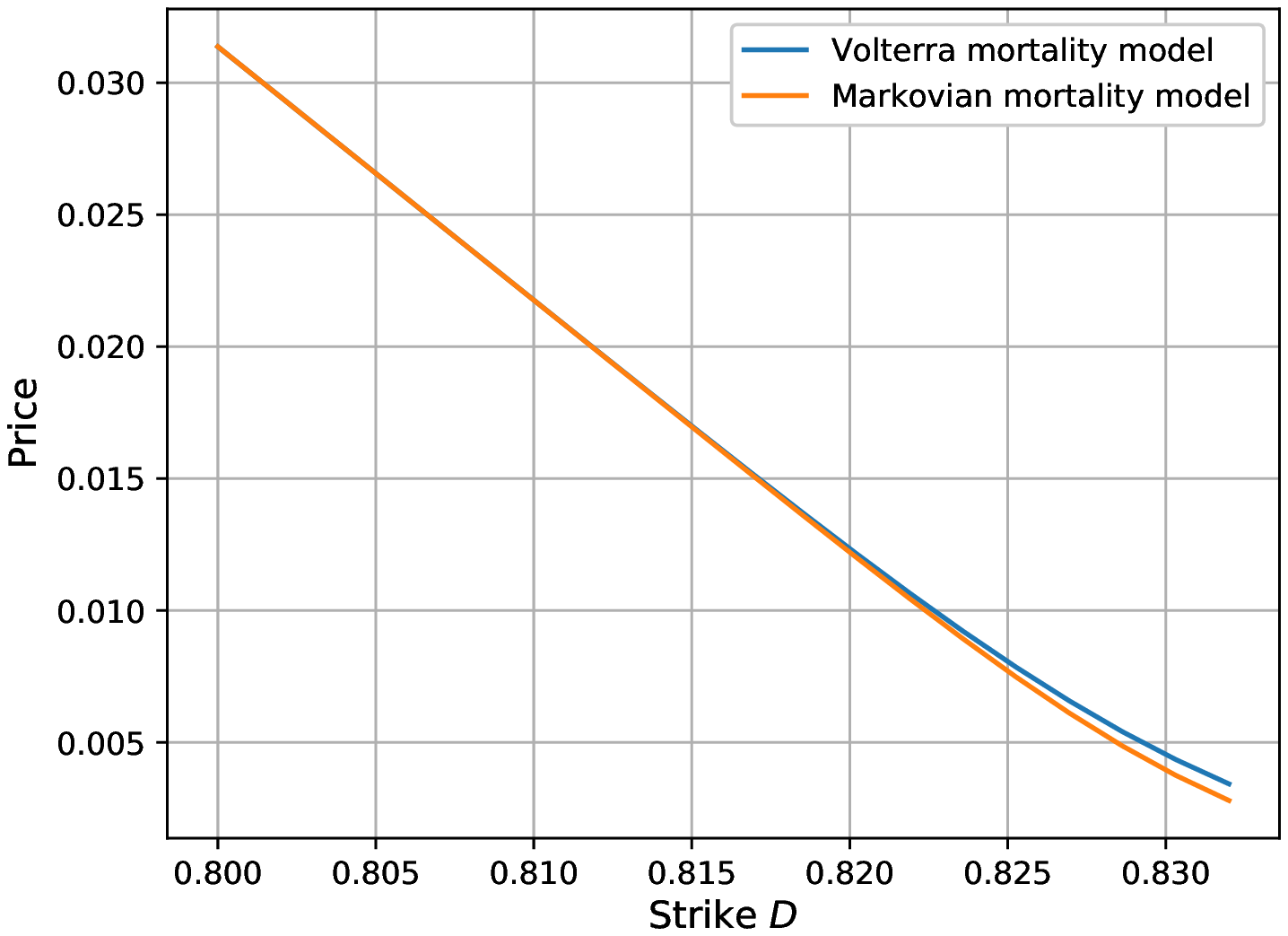}\label{fig:option1}}
	\subfigure[]{\includegraphics[width=7cm]{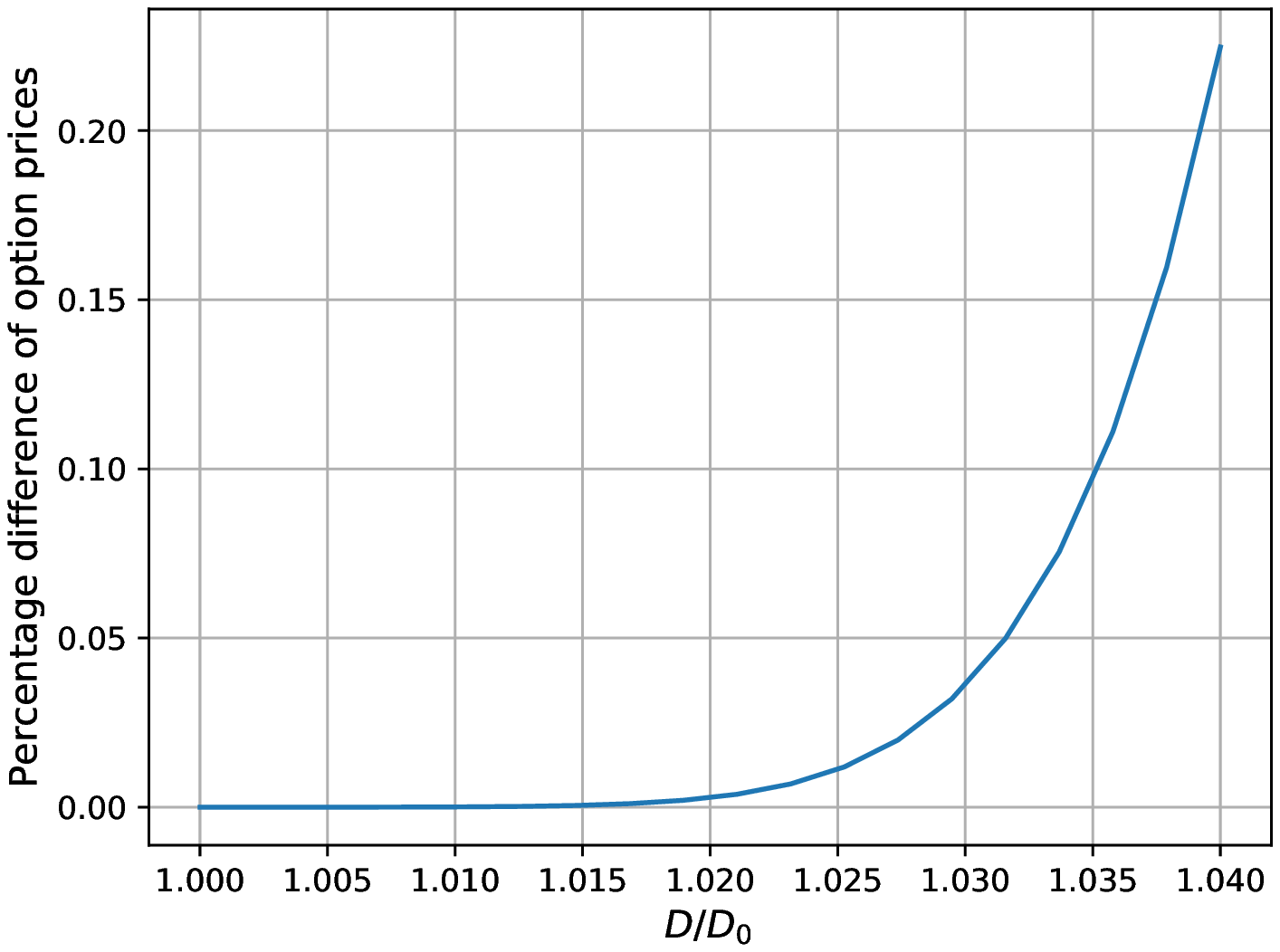}\label{fig:option2}}
	\caption{(a) Option prices and (b) difference of the prices under the two models}
	\label{fig:option}
\end{figure}

\subsection{LRD effect on longevity hedging}
We further examine the hedging with LRD. In this part, we still consider the fractional kernel in  \eqref{mortality} so that $K(t)=\frac{t^{\alpha-1}}{\Gamma(\alpha)}$.  Again, we first simulate a pair of sample paths of mortality and interest rates as shown in Figure \ref{sample paths}. The model parameters used are $\mu(0) = 0.15$, $b^1 = 0.5$, $b^0 = 0.1$, $\sigma_{\mu} = 0.05$, $r(0) = 0.04$, $\widetilde{b}^1$ = 0.6, $\widetilde{b}^0 = 0.02$, $\sigma_r = 0.01$, $T_0 = 5$, $\alpha = 1.33$, $k_1 = 1$, $k_2 = 10$, and $\mathbb{E}[z]=2$. 
\begin{figure}[H]
	\centering
	\subfigure[]{\includegraphics[width=7cm]{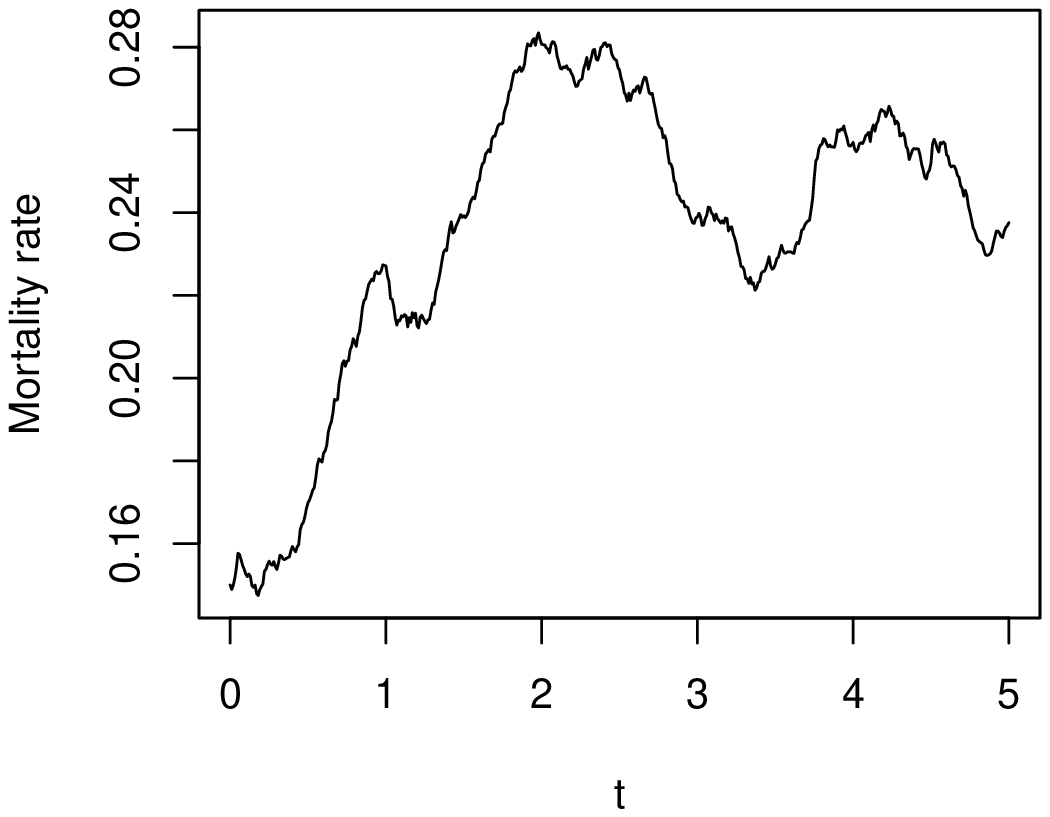}}
	\hspace{0ex}
	\subfigure[]{\includegraphics[width=7cm]{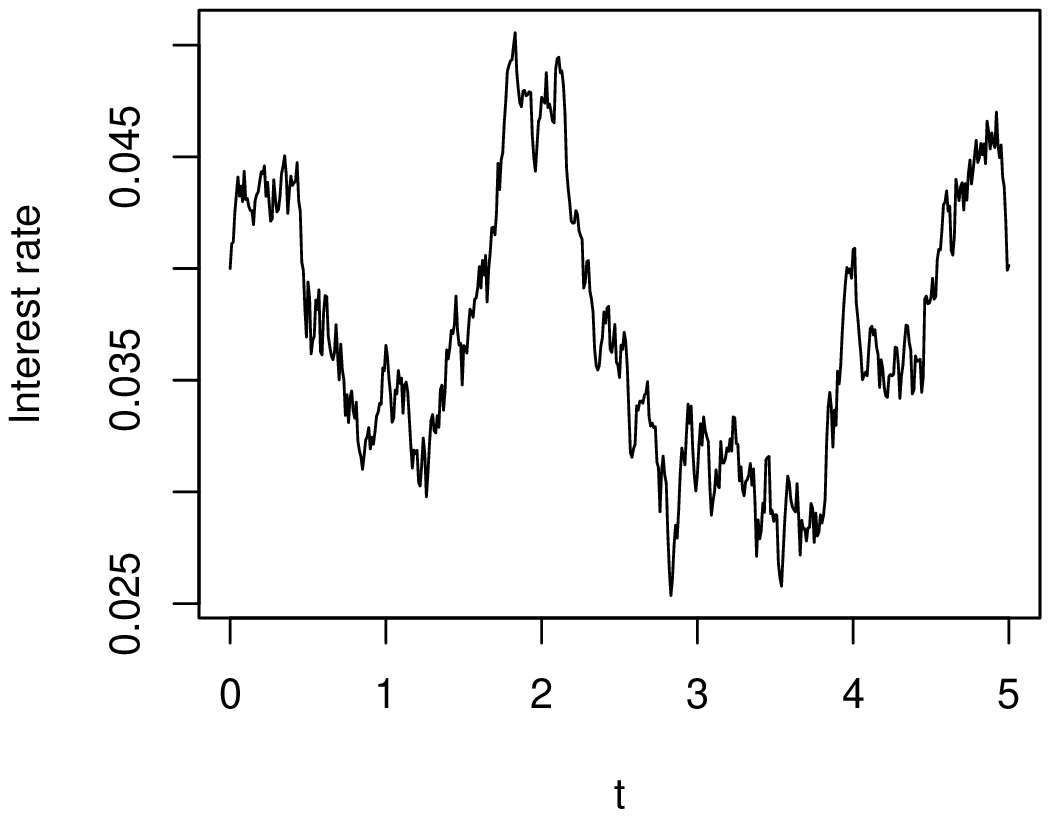}}
	\caption{ A pair of sample paths of (a) mortality rate and (b) interest rate}
	\label{sample paths}
\end{figure}

We hedge with the following two models. 
\begin{itemize}
\item[\textbullet]Model 1: Above assumption  with $K(t)=\frac{t^{\alpha-1}}{\Gamma(\alpha)}$ (Volterra mortality model); 
\item[\textbullet]Model 2: Above assumption  with $K(t)= 1$ (Markovian mortality model). 
\end{itemize}

Our objective is to hedge with $\phi = 3000$ over a horizon of 5 years using a zero-coupon longevity bond and a zero-coupon bond with a maturity time $T = 15$.  The initial value of wealth process is set to 2000. The optimal hedging strategies are calculated according to \eqref{u1} and \eqref{u2}. The longevity bond price and bond price are calculated by assuming constant market price of risks $\varphi = 0.1$ and $\vartheta = 0.1$. 

The optimal hedging strategies and corresponding wealth processes under the two models are plotted in Figures \ref{hedging strategy} and \ref{wealth:fig}, respectively. Once the mortality rate has the LRD feature, our hedging strategy significantly outperforms its Markovian counterpart and the unhedged position. Numerically, the objective  function value for Model 1 is -3622443 which is less than -3620889, the value for Model 2. As our goal is to minimize the MV objective, the smaller the number the better performance in terms of the objective function. If one is concerned about the risk level or the variance here, we report that the variance of the terminal wealth  is 66120 under Model 1 and 66317 under Model 2. The LRD hedging strategy prevails, too.
\textcolor{black}{We stress that this does not mean that the LRD hedging must be better in reality. Instead, we want to demonstrate the potential loss in hedging effectiveness with the Markovian model once the mortality rate has the LRD feature.}

\textcolor{black}{Although we set $\alpha = 1.33$ (or $H = 0.83$) in this numerical experiment,   the value of $\alpha$ can be calibrated or estimated in practice by using the pricing formulas we provide. Therefore, this study offers the option of choosing between Volterra and Markovian mortality models when dealing with longevity hedging in reality. Our proposed model renders a practical, flexible approach to the choice of $\alpha$. } 

\begin{figure}[H]
	\centering
	\subfigure[]{\includegraphics[width=7cm]{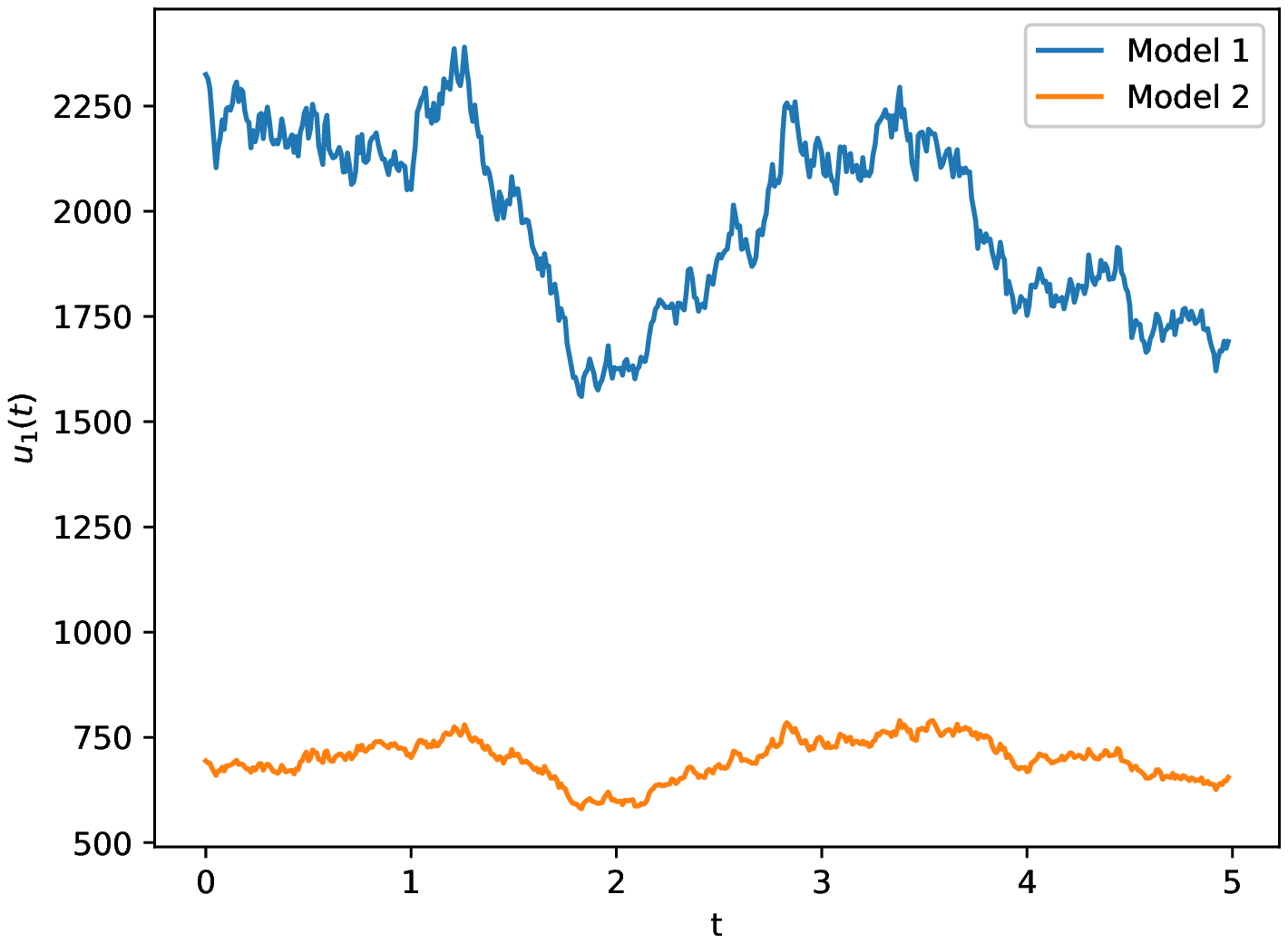}}
	\hspace{0ex}
	\subfigure[]{\includegraphics[width=7cm]{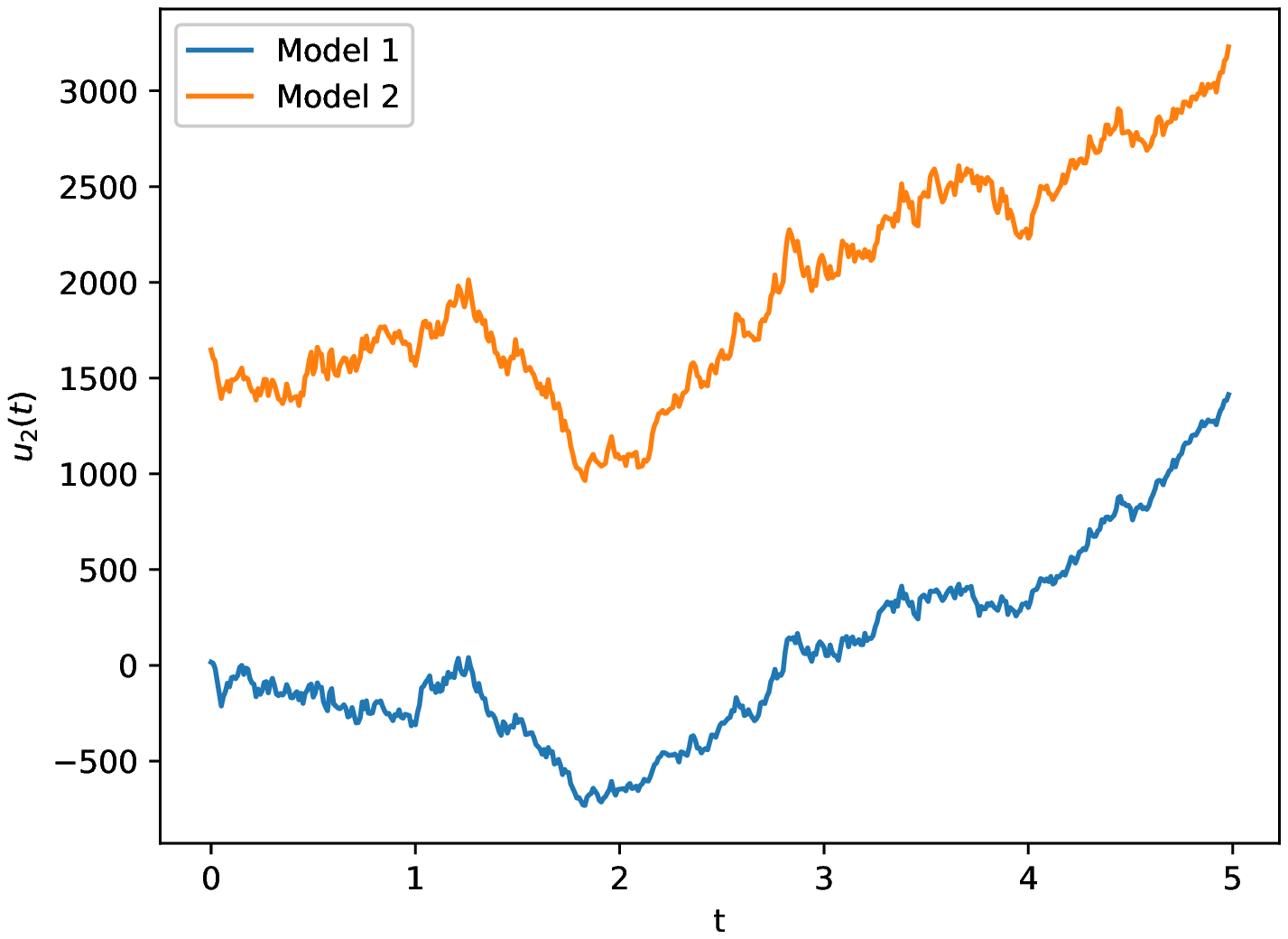}}
	\caption{Optimal hedging strategy (a) $u_1(t)$ and (b) $u_2(t)$ }
	\label{hedging strategy}
\end{figure}

\begin{figure}[H]
	\centering
	\includegraphics[width=9cm]{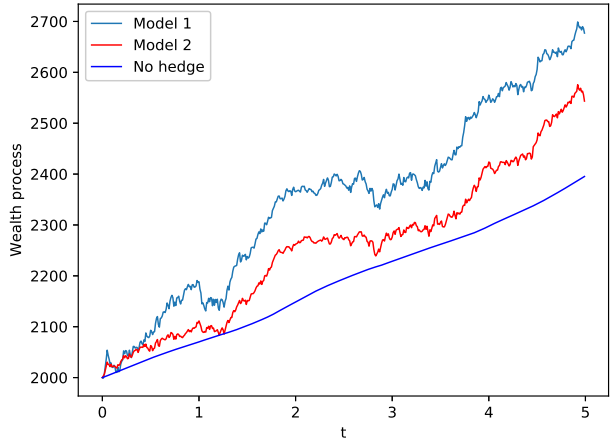}
	\caption{Wealth processes}
	\label{wealth:fig}
\end{figure}

\section{Conclusion}\label{Conclusion}
In this paper, we propose a tractable continuous-time mortality rate model that incorporates the LRD feature. Using our model, we derive novel closed-form solutions to the survival probability and prices of several basic insurance products. In addition, our model enables us to investigate an optimal longevity hedging strategy via the BSDE framework. Therefore, the key advantages of our model are its tractability for pricing and risk management as well as its ability to capture the LRD feature. Our numerical experiments show that LRD has significant effects for insurance pricing and hedging. The new longevity hedging strategy improves the hedging effectiveness when the mortality rate observes the LRD feature. 

\begin{appendices}
\addtocontents{toc}{\protect\setcounter{tocdepth}{0}}
\renewcommand{\appendixname}{Appendix~\Alph{section}}
\section{Transformation of Markov affine processes}
\label{appendix:affine}
We now give the ODEs which the coefficients $\widetilde{\alpha}$ and $\widetilde{\beta}$ solve appearing in Section 2 and 4.  
A $\mathbb{R}^k$-valued affine diffusion $Z$ is a $\mathbb{F}$-Markovian process specified as the strong solution to the following SDE:
\[dZ_t = \widetilde{b}(Z_t)dt +\widetilde{\sigma}(Z_t)dW'_t,\]
where $W'_t$  is a $\mathbb{F}$-standard $k$-dimensional Brownian motion. We require the covariance matrix $\widetilde{a}(Z)=\widetilde{\sigma}(Z)\widetilde{\sigma}(Z)^\top$ and the drift $\widetilde{b}(Z)$ to have affine dependence on $Z$ as in Definition \ref{def1}. That is
\begin{align}
\widetilde{a}(Z)&=\widetilde{A}^0+Z_1\widetilde{A}^1+\cdots+Z_k\widetilde{A}^k,\notag\\
\widetilde{b}(Z)&=\widetilde{b}^0+Z_1\widetilde{b}^1+\cdots+Z_k\widetilde{b}^k,\notag
\end{align}
for some $k$-dimensional symmetric matrices $\widetilde{A}^i$ and vectors $\widetilde{b}^i$. For convenience, we set $\widetilde{A}_1=(\widetilde{A}^1,\cdots,\widetilde{A}^k)$ and $\widetilde{b}_1=(\widetilde{b}^1,\cdots,\widetilde{b}^k)$. As shown in \cite{DPS}, for any $c_1,c_2\in \mathbb{C}^k$ and $c_3\in \mathbb{C}$, given $T>t$ and affine function $\Lambda(t,x) = \lambda_0(t) +\lambda_1(t)\cdot Z$ ($\lambda_0$ and $\lambda_1$ are  bounded continuous functions), under technical conditions we have
\[\mathbb{E}[e^{-\int_{t}^{T}\Lambda(s,Z_s)ds}e^{c_1\cdot Z_T}(c_2\cdot Z_T + c_3)|\mathcal{F}_t]= e^{\widetilde{\alpha}(t)+\widetilde{\beta}(t)\cdot Z_t}[\hat{\alpha}(t)+\hat{\beta}(t)\cdot Z_t],\]
where the functions $\widetilde{\alpha}(\cdot)\doteq\widetilde{\alpha}(\cdot, T)$ and $\widetilde{\beta}(\cdot)\doteq\widetilde{\beta}(\cdot, T)$ solve the following ODEs:
\[\dot{\widetilde{\beta}}(t)=\lambda_1(t)-\widetilde{b}_1(t)^\top\widetilde{\beta}(t)-\frac{1}{2}\widetilde{\beta}(t)^\top\widetilde{A}_1(t)\widetilde{\beta}(t),\]
\[\dot{\widetilde{\alpha}}(t)=\lambda_0(t)-\widetilde{b}^0(t)\cdot\widetilde{\beta}(t)-\frac{1}{2}\widetilde{\beta}(t)^\top\widetilde{A}^0(t)\widetilde{\beta}(t),\]
with boundary conditions $\widetilde{\alpha}(T)=0$ and $\widetilde{\beta}(T)=c_1$; the functions $\hat{\alpha}(\cdot)\doteq \hat{\alpha}(\cdot;c_1,c_2,c_3,T)$ and $\hat{\beta}(\cdot)\doteq \hat{\beta}(\cdot;c_1,c_2,c_3,T)$ are the solutions to the following ODEs:
\[\dot{\hat{\beta}}(t)=-\widetilde{b}_1(t)^\top\hat{\beta}(t)-\widetilde{\beta}(t)^\top\widetilde{A}_1(t)\hat{\beta}(t),\]
\[\dot{\hat{\alpha}}(t)=-\widetilde{b}^0(t)\cdot\hat{\beta}(t)-\widetilde{\beta}(t)^\top\widetilde{A}^0(t)\hat{\beta}(t),\]
with boundary conditions $\hat{\alpha}(T)=c_3$ and $\hat{\beta}(T)=c_2$.
\section{Some Proofs}
\label{appendix:proof}
\subsection*{Proof of Theorem \ref{theorem1}}
	Under our model, from \eqref{mu0}, 
	\[\mathbb{E}[e^{-\int_{t}^{T} \mu_sds}|\mathcal{F}_t] = \mathbb{E}[e^{-\int_{t}^{T}m(s)+\eta X_s}ds|\mathcal{F}_t] = e^{-\int_{t}^{T}m(s)ds}\mathbb{E}[e^{-\int_{t}^{T}\eta X_sds}|\mathcal{F}_t].\]
	As $X_t$ has the affine structure specified in Definition \ref{def1}, by application of Lemma 4.2 and  Theorem 4.3 provided in \cite{Jaber}, we have
	\[\mathbb{E}[e^{-\int_{t}^{T}\eta X_sds}|\mathcal{F}_t]=e^{\int_{0}^{t}\eta X_sds}\mathbb{E}[e^{-\int_{0}^{T}\eta X_sds}|\mathcal{F}_t] = e^{\int_{0}^{t}\eta X_sds}{\rm exp}(Y_t(T)),\]
	where $Y_t(T)$ is the Markovian process  defined in \eqref{Y} or equivalently \eqref{Y2} in Theorem \ref{theorem1}. Then, for $T>t\geq 0$, we have
	\begin{equation}
	\mathbb{E}[e^{-\int_{t}^{T} \mu_sds}|\mathcal{F}_t]= e^{-\int_{t}^{T}m(s)ds}e^{\int_{0}^{t}\eta X_sds}{\rm exp}(Y_t(T)).\nonumber
	\end{equation}
	Notice that $-\int_{t}^{T}m(s)ds+\int_0^t\eta X_sds = -\int_{0}^{T}m(s)ds+\int_0^t\mu_sds$. Hence, 
	\begin{equation}\label{g}
	\mathbb{E}[e^{-\int_{t}^{T} \mu_sds}|\mathcal{F}_t]= e^{-\int_{0}^{T}m(s)ds}e^{\int_{0}^{t}\mu_sds}{\rm exp}(Y_t(T))=g(t,T).
	\end{equation}
	By taking the derivative of $g(t, T)$ with respect to $T$, we get
	\begin{equation}\label{partial g}
	-\frac{\partial g(t,T)}{\partial T}=\mathbb{E}[e^{-\int_{t}^{T}\mu_sds}\mu_T|\mathcal{F}_t],~ T>t.
	\end{equation}
	Then, by combining the Equations \eqref{g} and \eqref{partial g}, the result in \eqref{eq} follows.
\subsection*{Proof of Theorem \ref{theorem 2} and Proposition \ref{proeta}} 
	For $P(t)$,  it is obvious that $P(t)>0$, $P(T_0) = 1$,  and
	\begin{align}
	\begin{split}
	P^{-1}(t) &= e^{\int_{t}^{T_0}\vartheta^2(s)+ \varphi^2(s)ds}\hat{\mathbb{E}}[e^{-2\int_{t}^{T_0}r(s)ds}|\mathcal{F}_t]\\
	&= e^{\int_{t}^{T_0}\vartheta^2(s)+ \varphi^2(s)ds}{\rm exp}(\alpha_1(t,T_0)+ \beta_1(t, T_0)r(t)).\notag
	\end{split}
	\end{align} 
	Under our setting, $\nu(t)^\top\Sigma(t)^{-1}\nu(t)= \vartheta^2(t) + \varphi^2(t)$.
    Then, by applying It\^o's formula, we get
	\begin{align}
	\begin{split}
	dP^{-1}(t)
	&= P^{-1}(t)(2r(t)-\vartheta^2(t) - \varphi^2(t))dt - P^{-1}(t)\widetilde{\eta}_1(t)^\top d\hat{\bm{W}}(t)\\
	&= P^{-1}(t)(2r(t)-\nu(t)^\top\Sigma(t)^{-1}\nu(t)-\tilde{\eta}_1(t)^\top\xi(t))dt - P^{-1}(t)\widetilde{\eta}_1(t)^\top d\bm{W}(t),  \notag
	\end{split}
	\end{align}
	where $\widetilde{\eta}_1 = -\beta_1(t, T_0)\sigma_r = \eta_1/P(t)$ and $\eta_1(t)$ is defined in Proposition \ref{proeta}. Notice that $\xi(t) = 2\sigma_S\Sigma(t)^{-1}\nu(t)$ and $\sigma^\perp\widetilde{\eta}_1= 0$. Then by It\^o's lemma again, $P(t)$ satisfies
	\begin{align}\label{BSDE:P}
	&dP(t)= \bigg\{\left[-2r(t) + \nu(t)^\top\Sigma(t)^{-1}\nu(t)\right]P(t) + 2\nu(t)^\top\Sigma(t)^{-1}\sigma_S(t)^\top\eta_1(t)\notag\\
	&+ \eta_1(t)^\top\sigma_S(t)\Sigma(t)^{-1}\sigma_S(t)^\top\eta_1(t)\frac{1}{P(t)}
	\bigg\}dt + \eta_1(t)^\top d\bm{W}(t).\notag
	\end{align}	
	For $Q(t)$,  it is obvious that $Q(T_0) = -c$ and $\frac{Q(t)}{P(t)} = -[Q_0(t) + c\acute{\mathcal{B}}(t, T_0)]$. 
	By applying  It\^o's lemma to $\acute{\mathbb{E}}[\mu_s|\tilde{\mathcal{H}}_t]$ on time $t$, we have
	\[d\left(\acute{\mathbb{E}}[\mu_s|\tilde{\mathcal{H}}_t]\right)= E_{B}(s-t)\sigma_{\mu}d\acute{W}_t,\]
    where $E_{B}$ is defined in Theorem \ref{theorem1} with $B = -b^1$.  By applying Ito's lemma to $\acute{\mathbb{E}}\left[\left.e^{-\int_{0}^{s}\mu_{\tau}d\tau}\right|\tilde{\mathcal{H}}_t\right] = {\rm exp}(Y_t^2(T))$ on time $t$, we get
	\[d\left(\acute{\mathbb{E}}\left[\left.e^{-\int_{0}^{s}\mu_{\tau}d\tau}\right|\tilde{\mathcal{H}}_t\right] \right)= \acute{\mathbb{E}}\left[\left.e^{-\int_{0}^{s}\mu_{\tau}d\tau}\right|\tilde{\mathcal{H}}_t\right]\psi_2(s-t)\sigma_{\mu}d\acute{W}_{t}\]
	with $\psi_2\in \mathcal{L}^2([0,s],\mathbb{C})$ solving the Riccati equation $\psi_2 = (-1 - \psi_2b^1)*K$.
	From \eqref{expr}, $d\acute{\mathcal{B}}(t,s)=\acute{\mathcal{B}}(t,s)r(t)dt - \acute{\mathcal{B}}(t,s)\beta_2(t,s)\sigma_r d\acute{W}'_t$. Then, by applying It\^o's lemma to $\frac{Q(t)}{P(t)}$, we have
	\begin{align}
	\begin{split}
	d\left[\frac{Q(t)}{P(t)}\right]&=\left[\frac{Q(t)}{P(t)}r(t) + k_1\mu(t)z + \pi(t)\right]dt +  \left[\widetilde{\eta}_2(t)^\top-\frac{Q(t)}{P(t)}\widetilde{\eta}_1(t)^\top\right]d\acute{\bm{W}}(t)\\
	&= \left[\frac{Q(t)}{P(t)}r(t) + k_1\mu(t)z + \pi(t) + \widetilde{\eta}_2(t)^\top\zeta(t)-\frac{Q(t)}{P(t)}\widetilde{\eta}_1(t)^\top\zeta(t)dt\right]\\
	&+ \left[\widetilde{\eta}_2(t)^\top-\frac{Q(t)}{P(t)}\widetilde{\eta}_1(t)^\top\right]d\bm{W}(t), \notag
	\end{split}
	\end{align}
	where $\widetilde{\eta}_2 = \eta_2/P(t)$ and $\eta_2$ is shown in Proposition \ref{proeta}. Notice that $\zeta(t) = \sigma_S\Sigma(t)^{-1}\nu(t)$ and $\sigma^\perp\widetilde{\eta}_1= 0$.  Then, by It\^o's lemma again, $Q(t)$ satisfies
	\begin{align}\label{BSDE:Q}
	&dQ(t) = \bigg\{\left[-r(t) + \nu(t)^\top\Sigma(t)^{-1}\left(\nu(t) + \frac{\sigma_S(t)^\top\eta_1(t)}{P(t)}\right)\right]Q(t) + P(t)(k_1\mu(t)z + \pi(t))\notag\\
	&+ \eta_2(t)^\top\sigma_S(t)\Sigma(t)^{-1}\left(\nu(t) + \frac{\sigma_S(t)^\top\eta_1(t)}{P(t)}\right)\bigg\}dt + \eta_2(t)^\top d\bm{W}(t).\notag
	\end{align}
	Finally, we consider the process $P(t)\left(M(t)+ \frac{Q(t)}{P(t)}\right)^2 + I(t)$.	By It\^o's formula, we have
	\begin{align}
	\begin{split}
	&d\left[P(t)\left(M(t)+ \frac{Q(t)}{P(t)}\right)^2 + I(t)\right] = d[P(t)M^2(t) + 2d[Q(t)M(t)] + d[Q^2(t)P^{-1}(t)] + dI(t)\\
	&= P(t)(u(t)-u_c^*(t))^\top\sigma_S(t)^\top\sigma_S(t)(u(t)-u_c^*(t))dt + \{\cdots\}d\bm{W}(t)+ \{\cdots\}d\mathcal{K}(t)\\
	&=P(t)||\sigma_S(t)(u(t)-u_c^*(t))||^2dt + \{\cdots\}d\bm{W}(t)+ \{\cdots\}d\mathcal{K}(t),\notag
	\end{split}
	\end{align}
	where $u^*_c(t)$ is defined in \eqref{ustar} and $\mathcal{K}(t) = \tilde{N}(t) - k_1\int_{0}^{t}\mu(s)ds$ is a martingale with respect to the filtration $\tilde{\mathcal{H}}_t$.  Then, there exists an increasing sequence of stopping times $\{\tau_i\}$ such that $\tau_i \uparrow T_0$ as $i \to \infty$ and 
	\begin{align}
	\begin{split}
	&\mathbb{E}\left[P(T_0\wedge\tau_i)\left(M(T_0\wedge\tau_i)+ \frac{Q(T_0\wedge\tau_i)}{P(T_0\wedge\tau_i)}\right)^2 + I(T_0\wedge\tau_i)\right]\\
	&= P(0)(Y(0)+\frac{Q(0)}{P(0)})^2 + I(0)+ \mathbb{E}\left[\int_{0}^{T_0\wedge\tau_i}P(t)||\sigma_S(t)(u(t)-u_c^*(t))||^2dt\right].  \notag
	\end{split}
	\end{align}
	From \eqref{P} and \eqref{Q}, we can see $P(t)$ and $Q(t)$ are bounded. From \eqref{I(t)}, $I(t)$ is also bounded.  As $\mathbb{E}[{\rm sup}_{t\in[0, T_0]}|Y^2(t)|^2]< \infty$,  according to the Dominance Covergence Theorem and Monotone Convergence Theorem as $i \to \infty$, we have
	\begin{align}
	\begin{split}
	&\mathbb{E}\left[P(T_0)\left(M(T_0)+ \frac{Q(T_0)}{P(T_0)}\right)^2 + I(T_0)\right]\\
	&=P(0)(Y(0)+\frac{Q(0)}{P(0)})^2 + I(0)+ \mathbb{E}\left[\int_{0}^{T_0}P(t)||\sigma_S(t)(u(t)-u_c^*(t))||^2dt\right].\notag
	\end{split}
	\end{align}
Thus, the objective function $\mathbb{E}[(M(T_0)-c)^2]= \mathbb{E}\left[P(T_0)\left(M(T_0)+ \frac{Q(T_0)}{P(T_0)}\right)^2 + I(T_0)\right]$  is minimized when $u(t) =u^*_t$. $P(0)(Y(0)+\frac{Q(0)}{P(0)})^2 + I(0)$ is the  optimal objective value.

\subsection*{Proof of Proposition \ref{pro:u}}
By Theorem \ref{theorem 2}, the optimal objective value is given by $P(0)(M(0)+ \frac{Q(0)}{P(0)})^2 + I(0)$  for any given $c$. Take $c = \bar{M} + \frac{\phi}{4}$ and substitute $Q(0) = -P(0)[Q_0(0)+c\acute{\mathcal{B}}(0,T_0)]$, then the external minimization problem in \eqref{dmin} becomes 
\[  \mathop{{\rm min}}\limits_{\bar{M}\in \mathbb{R}}P(0)(M(0)-( \bar{M} + \frac{\phi}{4})\acute{\mathcal{B}}(0,T_0)- Q_0(0))^2 + I(0) -\frac{\phi}{2}\bar{M}-\frac{\phi^2}{16},\]
which is a quadratic function attaining its minimum at 
\[\bar{M}^* = \frac{\frac{\phi}{4}(1-P(0)\acute{\mathcal{B}}^2(0,T_0))+P(0)\acute{\mathcal{B}}(0,T_0)(M(0)-Q_0(0))}{P(0)\acute{\mathcal{B}}^2(0,T_0)}.\]
By substituting $c = \bar{M}^* + \frac{\phi}{4}$, the result follows.
\end{appendices}


\end{document}